\documentclass{article}

\usepackage{ccfonts}
\usepackage{amssymb,amsmath,amsthm}
\usepackage{mathrsfs}
\usepackage{epsfig}

\newcommand\ie{{\em i.e.}~}
\newcommand\eg{{\em e.g.}~}

\def\D{\mathcal D}

\def\H{\mathcal H}

\def\m{\mathrm m}
\def\N{\mathbb N}

\def\R{\mathbb R}
\def\RR{\mathcal R}

\def\T{\mathsf T}

\def\lone{\mathsf{L}^{\:\!\!1}}

\def\cinf{\mathsf{C}^\infty}

\def\e{\mathop{\mathrm{e}}\nolimits}
\def\d{\mathrm{d}}

\def\supp{\mathop{\mathrm{supp}}\nolimits}
\def\Crit{\mathop{\mathsf{Crit}}\nolimits}

\def\Ran{\mathop{\mathsf{Ran}}\nolimits}
\def\Hess{\mathop{\mathsf{Hess}}\nolimits}

\newtheorem{Theorem}{Theorem}[section]
\newtheorem{Remark}[Theorem]{Remark}
\newtheorem{Lemma}[Theorem]{Lemma}
\newtheorem{Assumption}[Theorem]{Assumption}
\newtheorem{Corollary}[Theorem]{Corollary}
\newtheorem{Proposition}[Theorem]{Proposition}

\newtheorem{Example}[Theorem]{Example}

\begin{document}


\title{Time delay and Calabi invariant in classical scattering theory}

\author{A. Gournay$^1$ and R. Tiedra de Aldecoa$^2$\footnote{Supported by the
Fondecyt Grant 1090008 and by the Iniciativa Cientifica Milenio ICM P07-027-F
``Mathematical Theory of Quantum and Classical Magnetic Systems".}}

\date{\small}
\maketitle \vspace{-1cm}

\begin{quote}
\emph{
\begin{itemize}
\item[$^1$] Universit\'e de Neuch\^atel, Rue E.-Argand 11, CH-2000 Neuch\^atel,
Switzerland.
\item[$^2$] Facultad de Matem\'aticas, Pontificia Universidad Cat\'olica de
Chile,\\
Av. Vicu\~na Mackenna 4860, Santiago, Chile
\item[] \emph{E-mails:} antoine.gournay@unine.ch, rtiedra@mat.puc.cl
\end{itemize}
}
\end{quote}


\begin{abstract}
We define, prove the existence and obtain explicit expressions for classical time
delay defined in terms of sojourn times for abstract scattering pairs $(H_0,H)$ on a
symplectic manifold. As a by-product, we establish a classical version of the
Eisenbud-Wigner formula of quantum mechanics. Using recent results of V. Buslaev and
A. Pushnitski on the scattering matrix in Hamiltonian mechanics, we also obtain an
explicit expression for the derivative of the Calabi invariant of the Poincar\'e
scattering map.

Our results are applied to dispersive Hamiltonians, to a classical particle in a tube
and to Hamiltonians on the Poincar\'e ball.
\end{abstract}


\smallskip


\section{Introduction}\label{Intro}
\setcounter{equation}{0}

Since the works of D. Boll\'e, H. Narnhofer, W. Thirring and collaborators in the
80's, it is known that one can define properly a notion of time delay in terms of
sojourn times in classical scattering theory. However, most of the mathematical works
on the topic (if not all) provide a complete description only for scattering pairs
$(H_0,H)$ where the free Hamiltonian is of the type $H_0(q,p)=|p|^2/2$ on $\R^{2n}$.
Therefore, a legitimate interrogation is whether it is possible to define, to prove
the existence and to obtain explicit expressions for classical time delay for a
general class of scattering pairs in the modern set-up of symplectic geometry.
Answering (affirmatively) to these questions is the purpose of the present paper.

Our interest in these issues has been aroused by recent articles on time delay in
quantum mechanics \cite{RT10,RT11} and on the scattering matrix in Hamiltonian
mechanics \cite{BP09}. In \cite{RT11}, the authors prove that the existence of time
delay defined in terms of sojourn times, as well as its identity with Eisenbud-Wigner
time delay \cite{Smi60,Wig55}, is a common feature of two-Hilbert spaces quantum
scattering theory. Their proofs rely on abstract commutator methods and on an
integral formula relating localisation operators to time operators \cite{RT10}. Here,
we use the classical counterpart of this formula, established in \cite{GT11}, to
obtain similar results in classical scattering theory as well as an explicit
expression for the derivative of the Calabi invariant of the Poincar\'e scattering
map. Our approach takes its roots in the following observations\;\!: When
$H_0(q,p)=|p|^2/2$ on $\R^{2n}$, the usual position observables $\Phi_j(q,p):=q^j$
satisfy the simple Poisson bracket identity
\begin{equation}\label{Eq_bracket}
\big\{\{\Phi_j,H_0\},H_0\big\}=0.
\end{equation}
In consequence, the time evolution of the observables $\Phi_j$ under the flow
$\varphi^0_t$ of $H_0$ is lineal with growth rate $\{\Phi_j,H_0\}=p_j$. Accordingly,
any trajectory $\{\varphi^0_t(q,p)\}_{t\in\R}$ with initial velocity $p\neq0$ escapes
from each ball $B_r:=\{q\in\R^n\mid|q|\le r\}$ as $|t|\to\infty$. Similarly, if
$V:=H-H_0$ is a suitable perturbation of $H_0$ and if the initial condition $(q,p)$
is well chosen, the perturbed trajectory corresponding to the free trajectory
$\{\varphi^0_t(q,p)\}_{t\in\R}$ also escapes from each ball $B_r$ as $|t|\to\infty$.
In such a case, the difference of sojourn times in $B_r$ between the two trajectories
may converge to a finite value, called the global time delay for $(q,p)$, as
$r\to\infty$. This is well known and has been established by various authors for
different types of perturbations $V$ (see for instance
\cite{BD81,BGG83,BO81,GT07,KK92,Nar80,NT81,Thi97}). Fine. But what happens when $H_0$
and $H$ are abstract Hamiltonians on a given symplectic manifold $M$\;\!? If the
dimension of $M$ is finite, Darboux's theorem guarantees us that there exist, at
least locally, canonical coordinates on $M$. However, these coordinates have usually
nothing to do with the free Hamiltonian $H_0$, and are in consequence inappropriate
for the definition of sojourn times. Therefore, our point of view is instead to
retain as position observables merely functions $\Phi_j$ satisfying
\eqref{Eq_bracket}, as in the case $H_0(q,p)=|p|^2/2$. This choice is certainly not
the most general one, but it turns out to be extremely rewarding as we shall explain
below. Here, we just note three facts on its favour. First, it has been shown in
\cite[Sec.~4]{GT11} that there exist natural position observables $\Phi$ satisfying
\eqref{Eq_bracket} for many Hamiltonian systems $(M,H_0)$ appearing in literature.
Second, we know that this approach works in the quantum case \cite{RT11}. Finally,
the condition \eqref{Eq_bracket} is formulated in an invariant way on $M$, without
any mention to the particular structure of $H_0$.

So, let $H_0$ and $H$ be Hamiltonians on a symplectic manifold $M$ with Poisson
bracket $\{\;\!\cdot\;\!,\;\!\cdot\;\!\}$, assume that $H_0$ and $H$ have complete
flows $\{\varphi^0_t\}_{t\in\R}$ and $\{\varphi_t\}_{t\in\R}$, and let
$\Phi:=(\Phi_1,\ldots,\Phi_d)$ be a family of observables satisfying
\eqref{Eq_bracket}. Then, the vector
$\nabla H_0:=\big(\{\Phi_1,H_0\},\ldots,\{\Phi_d,H_0\}\big)$ and the set
$$
\Crit(H_0,\Phi):=(\nabla H_0)^{-1}(\{0\})\subset M
$$
can be interpreted, respectively, as the velocity observable and the set of critical
points associated to $H_0$ and $\Phi$ (see \cite[Ass.~2.2 \& Def.~2.5]{RT10} for
quantum analogues). Accordingly, the free trajectories $\{\varphi^0_t(m)\}_{t\in\R}$
with $m\in M\setminus\Crit(H_0,\Phi)$ escape from each set $\Phi^{-1}(B_r)$ as
$|t|\to\infty$ (as in the case of $H_0(q,p)=|p|^2/2$, where
$\Crit(H_0,\Phi)=\R^n\times\{0\}$ and $\Phi^{-1}(B_r)=B_r$). Therefore, we have
propagation in $M\setminus\Crit(H_0,\Phi)$, and the wave maps 
$W_\pm:=\lim_{t\to\pm\infty}\varphi_{-t}\circ\varphi_t^0$ exist and are well defined
symplectomorphisms on $M\setminus\Crit(H_0,\Phi)$ if $V\equiv H-H_0$ is suitable
(Theorem \ref{t-thithi}). Using a virial type argument, we then show in Lemma
\ref{t-condvir} that $W_\pm$ are complete, and thus that the scattering map
$S:=W_+^{-1}\circ W_-$ is also a well defined symplectomorphism on
$M\setminus\Crit(H_0,\Phi)$. With these objects at hand, we introduce in Section
\ref{Sec_time} the symmetrised time delay $\tau_r$, defined in terms of sojourn times
in the sets $\Phi^{-1}(B_r)$, for the general scattering system $(M,H_0,H)$. Then, we
prove the existence of the limit $\tau:=\lim_{r\to\infty}\tau_r$ and its identity
with a difference of arrival times similar to Eisenbud-Wigner Formula in quantum
mechanics (Theorem \ref{Teo_sym} and Remark \ref{Rem_EW}). We also show in Corollary
\ref{Cor_unsym} that the usual (unsymmetrised) time delay exists and is equal to the
symmetrised time delay if the scattering process preserves the norm of the velocity
vector $\nabla H_0$. Finally, we establish in Section \ref{Sec_Cal} a link between
our results on the whole manifold $M$ and the results of \cite{BP09} on fixed energy
submanifolds $\Sigma_E^0:=H_0^{-1}(\{E\})$. Under appropriate assumptions on the
energy $E\in\R$ and the perturbation $V$, we show that the abstract time delay
$\tau_E$ defined in \cite{BP09} as the difference of distance from a Poincar\'e
section $\Gamma_E\subset\Sigma_E^0$ before and after scattering coincides with the
restriction $\tau|_{\Gamma_E}$ if
$$
\Gamma_E=\big\{m\in\Sigma^0_E\mid(\Phi\cdot\nabla H_0)(m)=0\big\}.
$$
In other terms, if there exist position observables $\Phi$ satisfying
\eqref{Eq_bracket}, then there exist natural Poincar\'e sections $\Gamma_E$ verifying
the assumptions of \cite{BP09}, and our general time-dependent definition of time
delay coincides, after restriction to $\Gamma_E$, with the abstract time-independent
definition of time delay of \cite{BP09}. This establishes a new relation between two
complementary formulations of classical scattering theory. Furthermore, by using a
theorem of \cite{BP09} linking $\tau_E$ to the Calabi invariant
$\mathrm{Cal}\big(\widetilde S_E\big)$ of the Poincar\'e scattering map
$\widetilde S_E$, this leads to an explicit expression in terms of $\Phi$, $S$ and
$\nabla H_0$ for the derivative
$\frac\d{\d E}\;\!\mathrm{Cal}\big(\widetilde S_E\big)$ of the Calabi invariant
(Theorem \ref{Thm_Cal}).

To conclude, we point out several aspects of independent interest in the paper. In
the first place, our results allow to treat three new classes of Hamiltonians
systems\;\!: dispersive Hamiltonians $H_0(q,p)=h(p)$ on $\R^{2n}$, a classical
particle in a tube and the kinetic energy Hamiltonian on the Poincar\'e ball. The
first example generalises the case $H_0(q,p)=|p|^2/2$, the second example illustrates
the fact that in general only the symmetrised time delay exists, and the last example
shows that our results also apply to geodesic flows on manifolds with curvature.
We also note that our treatment of classical scattering theory in Section
\ref{Sec_Clas_Scat} is model-independent, and therefore possibly of some
use in other contexts. Finally, we note that the present paper provides a new
example of results valid both in quantum and classical mechanics. Accordingly, we try
to put into light throughout all of the paper the relation between both theories.

\section{Classical scattering theory}\label{Sec_Clas_Scat}
\setcounter{equation}{0}

In this section, we introduce a classical scattering pair $(H_0,H)$ on a symplectic
manifold $M$ and a family of position type observables
$\Phi\equiv(\Phi_1,\ldots,\Phi_d)$ satisfying the Poisson bracket relation \eqref{Eq_bracket}. Then, we recall some results on $H_0$ following from the existence of the
family $\Phi$. Finally, we extend standard results on the scattering theory for
$H_0(q,p)=|p|^2/2$, $H(q,p)=|p|^2/2+V(q)$ and $\Phi(q,p)=q$ to the abstract triple
$(H_0,H,\Phi)$.

Throughout the paper, we use the notations $\R_+:=(0,\infty)$ and $\R_-:=(-\infty,0)$, and we write $\zeta^*$ for the pullback of a
diffeomorphism of manifolds $\zeta:M_1\to M_2$.

\subsection{Free Hamiltonian and position observables}\label{Sec_Free_Ham}

Let $M$ be a symplectic manifold, \ie a smooth manifold endowed with a closed
two-form $\omega$ such that the morphism
$
TM\ni X\mapsto\omega^\flat(X):=\iota_X\omega
$
is an isomorphism. In infinite dimension, such a manifold is said to be a strong
symplectic manifold (in opposition to a weak symplectic manifold, when the above map
is only injective; see \cite[Sec.~8.1]{AMR88}). When the dimension is finite, the
dimension must be even, say equal to $2n$, and the $2n$-form
$\omega^n:=\omega\wedge\cdots\wedge\omega$ must be a volume form. The Poisson bracket
is defined as follows: for each $f\in\cinf(M)$ we define the vector field
$X_f:=(\omega^\flat)^{-1}(\d f)$, \ie
$\d f(\;\!\cdot\;\!)=\omega(X_f,\;\!\cdot\;\!)$, and set $\{f,g\}:=\omega(X_f,X_g)$
for each $f,g\in\cinf(M)$.

In the sequel, the function $H_0\in\cinf(M)$ is an Hamiltonian with complete vector
field $X_{H_0}$. So, the flow $\{\varphi^0_t\}$ associated to $H_0$ is defined for
all $t\in\R$, it preserves the Poisson bracket, \ie
$$
\big\{f\circ\varphi^0_t,g\circ\varphi^0_t\big\}=\{f,g\}\circ\varphi^0_t,\quad
t\in\R,
$$
and it satisfies the usual evolution equation\;\!:
$$
\frac\d{\d t}\;\!f\circ\varphi^0_t=\big\{f,H_0\big\}\circ\varphi^0_t,\quad t\in\R.
$$
In particular, the Hamiltonian $H_0$ is preserved along its flow, \ie
$$
H_0\circ\varphi^0_t=H_0,\quad t\in\R.
$$

As in \cite[Sec.~3]{GT11}, we consider an additional family  
$\Phi\equiv(\Phi_1,\ldots,\Phi_d)\in\cinf(M;\R^d)$ of observables, and define
the
associated functions
$$
\partial_jH_0:=\{\Phi_j,H_0\}\in\cinf(M)\qquad\hbox{and}\qquad
\nabla
H_0\equiv\{\Phi,H_0\}:=(\partial_1H_0,\ldots,\partial_dH_0)\in\cinf(M;\R^d),
$$
and the corresponding set of critical points\;\!:
$$
\Crit(H_0,\Phi):=(\nabla H_0)^{-1}(\{0\})\subset M.
$$
The set $\Crit(H_0,\Phi)$ is closed and contains the usual set $\Crit(H_0)$ of
critical points of $H_0$, \ie
$$
\Crit(H_0,\Phi)
\supset\Crit(H_0)
:=\big\{m\in M\mid X_{H_0}(m)=0\big\}
\equiv\big\{m\in M\mid(\d H_0)_m=0\big\}.
$$

Our first assumption is the following\;\!:

\begin{Assumption}[Position observables]\label{AssCom}
One has $\big\{\{\Phi_j,H_0\},H_0\big\}=0\,$ for each $j\in\{1,\ldots,d\}$.
\end{Assumption}

Assumption \ref{AssCom} is verified by many free Hamiltonian systems $(M,\omega,H_0)$
appearing in literature (see \cite[Sec.~4]{GT11} for both finite and infinite
dimensional examples). It implies that the time evolution of the observables $\Phi_j$
under the flow $\{\varphi^0_t\}_{t\in\R}$ is lineal with growth rate $\partial_jH_0$;
namely,
\begin{equation}\label{Eq_linear}
\big(\Phi_j\circ\varphi^0_t\big)(m)=\Phi_j(m)+t\big(\partial_jH_0\big)(m)
\quad\hbox{for all $j\in\{1,\ldots,d\}$, $t\in\R$ and $m\in M$.}
\end{equation}
Furthermore, one has for each $t\in\R$
\begin{equation}\label{Eq_linear_bis}
\varphi^0_t\big(\Crit(H_0,\Phi)\big)=\Crit(H_0,\Phi)
\qquad\hbox{and}\qquad
\varphi^0_t\big(M\setminus\Crit(H_0,\Phi)\big)=M\setminus\Crit(H_0,\Phi),
\end{equation}
and if $m\in M\setminus\Crit(H_0,\Phi)$, then one must have $\varphi^0_t(m)\neq m$
for all $t\neq0$ due to Equation \eqref{Eq_linear}. So, each orbit
$\{\varphi^0_t(m)\}_{t\in\R}$ either stays in $\Crit(H_0,\Phi)$ if
$m\in\Crit(H_0,\Phi)$, or stays outside $\Crit(H_0,\Phi)$ and is not periodic if
$m\notin\Crit(H_0,\Phi)$.

Assumption \ref{AssCom} also permits to relate the difference of the sojourn
times (in
the past and in the future) of a classical orbit $\{\varphi^0_t(m)\}_{t\in\R}$
in the
dilated regions $\Phi^{-1}(B_r)\subset M$, $r>0$, to a finite arrival time
defined in
terms of $\Phi$ and $H_0$. To see this, let $T:M\setminus\Crit(H_0,\Phi)\to\R$
be the
$\cinf$-function given by
\begin{equation}\label{def_T}
T:=\Phi\cdot\frac{\nabla H_0}{|\nabla H_0|^2}\;\!,
\end{equation}
and then observe the following (see \cite[Sec.~3.1-3.2]{GT11} for details)\;\!:
\begin{enumerate}
\item[(i)] One has for each $t\in\R$
$$
\big\{T,H_0\big\}\circ\varphi^0_t\equiv\frac\d{\d t}\big(T\circ\varphi^0_t\big)=1
\qquad\hbox{and}\qquad
T\circ\varphi^0_t=T+t
$$
on $M\setminus\Crit(H_0,\Phi)$.
\item[(ii)] If we consider the observables $\Phi_j$ as the components of an abstract
position observable $\Phi$, then $\nabla H_0$ can be seen as the velocity vector for
the Hamiltonian $H_0$, and $-T(m)$ is equal to the time at which a particle in $\R^n$
with initial position $\Phi(m)$ and velocity $(\nabla H_0)(m)$ intersects the
hyperplane (containing the origin) orthogonal to the unit vector
$\frac{(\nabla H_0)(m)}{|(\nabla H_0)(m)|}$\;\!. For instance, if $\Phi(q,p)=q$ and
$H_0(q,p)=|p|^2/2$ are the usual position and kinetic energy on $M=T^*\R^n$, then
$-T(q,p)\equiv-q\cdot p/|p|^2$ is known in physics literature as the arrival time of
the free particle (see \eg \cite[Sec.~II.E]{Sas10}).
\end{enumerate}

Accordingly, the observable $T$ represents a time of arrival growing linearly under
the flow $\{\varphi^0_t\}_{t\in\R}$ (in quantum mechanics, the relations
$\{T,H_0\}=1$ and  $T\circ\varphi^0_t=T+t$ are replaced by the canonical commutation
relation $[T,H_0]=i$ and the Weyl relation $\e^{itH_0}T\e^{-itH_0}=T+t$, and the
corresponding operator $T$ is called a time operator for $H_0$; see \cite{Ara05} or
\cite{Miy01} for details). It only remains to link the observable $T$ to the sojourn
times of the classical orbits in the regions $\Phi^{-1}(B_r)$. This is the content of
the next theorem, proved in Section 3.2 of \cite{GT11}\;\!:

\begin{Theorem}\label{Cor_car}
Let $H_0$ and $\Phi$ satisfy Assumption \ref{AssCom}. Then we have
$$
\lim_{r\to\infty}\frac12\int_0^\infty\d t\,
\big\{\big(\chi_r^\Phi\circ\varphi_{-t}^0\big)(m)
-\big(\chi_r^\Phi\circ\varphi_t^0\big)(m)\big\}
=
\begin{cases}
T(m) & \hbox{if }\,m\in M\setminus\Crit(H_0,\Phi)\\
0 & \hbox{if }\,m\in\Crit(H_0,\Phi),
\end{cases}
$$
with $\chi_r^\Phi$ the characteristic function for the set $\,\Phi^{-1}(B_r)$.
\end{Theorem}

We conclude the section by exhibiting three examples which will serve as connecting
thread throughout the whole paper. However, we note that many other examples are
certainly accessible as suggested by \cite[Sec.~4]{GT11} (for instance, it would be
interesting to study the cases of the nonlinear Klein-Gordon and Schr\"odinger
equations where the scattering theory is well defined, see \cite{Nak99,Nak01}). The
notation $\langle y\rangle:=\sqrt{1+|y|^2}$ is used for any $y\in\R^n$.

\begin{Example}[$H_0(q,p)=h(p)$]\label{Ex_hp_1}
Consider on $M:=T^*\R^n\simeq\R^{2n}$ the canonical coordinates $(q,p)$, with
$q\equiv(q^1,\ldots,q^n)$ and $p\equiv(p_1,\ldots,p_n)$, and the canonical
symplectic
form $\omega:=\sum_{j=1}^n\d q^j\wedge\d p_j$. Take a purely kinetic Hamiltonian
$H_0(q,p):=h(p)$ with $h\in\cinf(\R^n;\R)$, and let $\Phi_j(q,p):=q^j$ be the
usual
position functions. Then $\varphi^0_t(q,p)=\big(q+t(\nabla h)(p),p\big)$,
$\nabla H_0=\nabla h$, and Assumption \ref{AssCom} is satisfied:
$$
\big\{\{\Phi_j,H_0\},H_0\big\}(q,p)
=\big\{(\partial_jh)(p),h(p)\big\}
=0.
$$
Furthermore, we have
$\Crit(H_0)=\Crit(H_0,\Phi)=\R^n\times(\nabla h)^{-1}(\{0\})$.
\end{Example}

\begin{Example}[Particle in a tube]\label{Ex_tube_1}
Let $\Omega:=\R\times\mathring B_1$ be a straight tube with section
$\mathring B_1:=\big\{q_\perp\in\R^n\mid|q|<1\big\}$, and endow
$M:=T^*\Omega\simeq\Omega\times\R^{n+1}$ with the coordinates
$q\equiv(q^1,q_\perp)\in\Omega$ and $p\equiv(p_1,p_\perp)\in\R^{n+1}$ and with the
symplectic form $\omega:=\sum_{j=1}^{n+1}\d q^j\wedge\d p_j$. Then, consider the
Hamiltonian given by the map
$$
H_0:M\to\R,\quad(q,p)\mapsto|p|^2/2+v_0\big(|q_\perp|^2\big),
$$
where $v_0\in\cinf\big((0,1)\big)$ satisfies
\begin{enumerate}
\item[(i)] $v_0\equiv0$ in a neighbourhood of $\,0\in\R$,
\item[(ii)] $v_0'(x)\ge0$ for all $x\in(0,1)$,
\item[(iii)] $\lim_{x\nearrow1}v_0(x)=+\infty$.
\end{enumerate}
The Hamiltonian $H_0$ models a classical particle in the tube $\Omega$ evolving under
the influence of a confining potential $v_0$, which repels the particle near the
boundary of $\Omega$. The motion of the particle is uniform along the $q^1$-axis and
given by the equation $\ddot q_\perp=-2q_\perp v_0'\big(|q_\perp|^2\big)$ in the
transverse direction. To show the completeness of the corresponding Hamiltonian
vector field
$$
X_{H_0}(q,p)
=p\;\!\frac\partial{\partial q}\bigg|_{(q,p)}
-2q_\perp v_0'\big(|q_\perp|^2\big)\frac\partial{\partial p_\perp}\bigg|_{(q,p)},
\quad(q,p)\in M,
$$
one can for instance use the criterion \cite[Prop.~2.1.20]{AM78} with the (proper and
$\cinf$\!) function
\begin{equation}\label{Func_tube}
f:M\to\R,\quad(q,p)\mapsto H_0(q,p)+\big\langle q^1\big\rangle.
\end{equation}
As a function $\Phi$, we take the longitudinal coordinate $\Phi(q,p):=q^1$. This gives
$(\nabla H_0)(q,p)=p_1$ as velocity observable and implies that
$$
\big\{\{\Phi,H_0\},H_0\big\}(q,p)
=\big\{p_1,|p|^2/2+v_0\big(|q_\perp|^2\big)\big\}
=0.
$$
So, Assumption \ref{AssCom} is satisfied and
$$
\Crit(H_0)
=\big(\R\times(v_0')^{-1}(\{0\})\big)\times\{0\}
\subset\Omega\times\big(\{0\}\times\R^n\big)
=\Crit(H_0,\Phi).
$$
\end{Example}

\begin{Example}[Poincar\'e ball]\label{Ex_Poin_1}
Consider the open unit ball $\mathring B_1\equiv\big\{q\in\R^n\mid|q|<1\big\}$
endowed with the Riemannian metric $g$ given by
$$
g_q(X_q,Y_q):=\frac4{(1-|q|^2)^2}\;\!(X_q\cdot Y_q),
\quad q\in\mathring B_1,~X_q,Y_q\in T_q\mathring B_1\simeq\R^n.
$$
Let $T^*\mathring B_1\simeq\big\{(q,p)\in\mathring B_1\times\R^n\big\}$ be the
cotangent bundle on $\mathring B_1$ with symplectic form
$\omega:=\sum_{j=1}^n\d q^j\wedge\d p_j$, and let
$$
H_0:T^*\mathring B_1\to\R,\quad(q,p)\mapsto\frac12\sum_{j,k=1}^ng^{jk}(q)p_jp_k
=\frac18\;\!|p|^2\big(1-|q|^2\big)^2
$$
be the kinetic energy Hamiltonian. Then, we know from \cite[Sec.~4.3(D)]{GT11} that
the Hamiltonian vector field $X_{H_0}$ is complete on
$
M:=T^*\mathring B_1\setminus H_0^{-1}(\{0\})
\simeq\mathring B_1\times\R^n\setminus\{0\}
$
and that the function
$$
\Phi:M\to\R,\quad(q,p)\mapsto
\tanh^{-1}\left(\frac{2(p\cdot q)}{|p|(1+|q|^2)}\right)
$$
is $\cinf$ and satisfies Assumption \ref{AssCom} with $\nabla H_0=\sqrt{2H_0}$ and
$\Crit(H)=\Crit(H,\Phi)=\varnothing$.

The observable $\Phi$ can be interpreted (\;\!\eg using isometries) as follows. Let
$\widetilde q$ be the closest point to the origin $0\in\mathring B_1$ on the geodesic
curve generated by $(q,p)$. Then $\Phi(q,p)$ is the geodesic distance between $q$ and
$\widetilde q$ together with a sign (positive if going from $\widetilde q$ to $q$
goes in the same direction as $p$ and negative otherwise).
\end{Example}

\subsection{Wave maps and scattering map}\label{Sec_wave}

From now on, we also consider a perturbed Hamiltonian $H\in\cinf(M)$ with complete
flow $\{\varphi_t\}_{t\in\R}$, and suppose for a moment that the pair $(H_0,H)$ is
such that the wave maps exist, have common ranges and are invertible\;\!:

\begin{Assumption}[Wave maps]\label{Ass_wave}
\begin{enumerate}
\item[]
\item[(i)] The pointwise limits
$W_\pm:=\lim_{t\to\pm\infty}\varphi_{-t}\circ\varphi_t^0$ exist on some open
sets
$\D_\pm\subset M$.
\item[(ii)] The maps $W_\pm$ are invertible, with inverses
$W_\pm^{-1}:\Ran(W_\pm) \to\D_\pm$.
\item[(iii)] The maps $W_\pm$ have common ranges equal to $\RR$, \ie
$\Ran(W_+)=\Ran(W_-)=\RR$.
\end{enumerate}
\end{Assumption}

The functions $W_\pm:\D_\pm\to\RR$ are called the wave maps and the condition
(ii)
is often referred as completeness of the wave maps. Since the flows
$\{\varphi_t^0\}_{t\in\R}$ and $\{\varphi_t\}_{t\in\R}$ are groups of
diffeomorphisms
on $M$, the wave maps $W_\pm$ verify the intertwining property
\begin{equation}\label{inter}
\big(\varphi_t\circ W_\pm\big)(m_\pm)=\big(W_\pm\circ\varphi_t^0\big)(m_\pm)
\quad\hbox{for all }t\in\R\hbox{ and }m_\pm\in\D_\pm.
\end{equation}
Due to points (ii) and (iii) of Assumption \ref{Ass_wave} the scattering map
$$
S:=W_+^{-1}\circ W_-
$$
is well defined and invertible from $\D_-$ to $\D_+$. Furthermore, the intertwining
property \eqref{inter} implies that $S$ commutes with the free evolution\;\!:
\begin{equation}\label{commutation_S}
\big(\varphi_t^0\circ S\big)(m_-)=\big(S\circ\varphi_t^0\big)(m_-)
\quad\hbox{for all }t\in\R\hbox{ and }m_-\in\D_-.
\end{equation}

In finite dimension, what precedes admits an interesting formulation (borrowed from
\cite[Sec.~2.3]{BP09}) on submanifolds of fixed energy. To see this, let $E\in\R$ be
a regular value of $H_0$, \ie $H_0^{-1}(\{E\})\cap\Crit(H_0)=\varnothing$. Then,
$\Sigma^0_E:=H_0^{-1}(\{E\})$ is a regular submanifold of $M$ of dimension $2n-1$,
the family $\big\{\varphi^{0,E}_t:=\varphi^0_t|_{\Sigma^0_E}\big\}_{t\in\R}$ is a
group of diffeomorphisms on $\Sigma^0_E$ and the group action
$$
\varphi^{0,E}:\R\times\Sigma^0_E\to\Sigma^0_E,~~(t,m)\mapsto\varphi^{0,E}_t(m)
$$
is smooth. Now, if $\varphi^{0,E}$ is free and proper, then the quotient (orbit)
space $\widetilde\Sigma^0_E:=\Sigma^0_E/\R$ is a smooth manifold of dimension
$2(n-1)$ and the quotient map $\pi^0_E:\Sigma^0_E\to\widetilde\Sigma^0_E$ is a
submersion \cite[Prop.~4.1.23]{AM78}. Furthermore, there exists a unique symplectic
form $\widetilde\omega^0_E$ on $\widetilde\Sigma^0_E$ such that 
$(\pi^0_E)^*\widetilde\omega^0_E=\omega|_{\Sigma^0_E}$ (see \cite[Thm.~4.3.1 \&
Ex.~4.3.4(ii)]{AM78}). If the situation is favourable enough, it may also happen that
$\Sigma^0_E\subset\D_-$ and that the restriction $S_E:=S|_{\Sigma^0_E}$ is a
diffeomorphism from $\Sigma^0_E$ onto $\Sigma^0_E$ such that
\begin{equation}\label{S_E_rest}
S_E^*\big(\omega|_{\Sigma^0_E}\big)=\omega|_{\Sigma^0_E}.
\end{equation}
In such a case, the map $S_E:\Sigma^0_E\to\Sigma^0_E$ is equivariant with respect to
the action $\varphi^{0,E}$ due to \eqref{commutation_S}, and thus induces a
diffeomorphism $\widetilde S_E:\widetilde\Sigma^0_E\to\widetilde\Sigma^0_E$ defined
by the relation
\begin{equation}\label{S_E_diagram}
\pi_E^0\circ S_E=\widetilde S_E\circ\pi_E^0.
\end{equation}
Furthermore, one obtains from \eqref{S_E_rest} and \eqref{S_E_diagram} that
$$
(\pi_E^0)^*\big((\widetilde
S_E)^*\widetilde\omega^0_E-\widetilde\omega^0_E\big)=0,
$$
meaning that $\widetilde S_E^*\widetilde\omega^0_E=\widetilde\omega^0_E$, since
$\pi_E^0$ is a surjection. This implies that $\widetilde S_E$ is a symplectomorphism
of $\big(\widetilde\Sigma^0_E,\widetilde\omega^0_E\big)$. In the case $n=1$, the
above reduction leads to a manifold $\widetilde\Sigma^0_E$ of dimension zero, \ie a
discrete set of orbits. So, the map $\widetilde S_E$ reduces to a permutation map on
the discrete set $\widetilde\Sigma^0_E$.

The map $\widetilde S_E$, called the Poincar\'e scattering map, will play an
important role in Section \ref{Sec_Cal}. We refer to \cite{DS92,Jun86,RS95} and
\cite{Kna99} for early works involving the Poincar\'e scattering map in physics and
mathematics literature respectively.

\subsection{Completeness of the wave maps}

Conditions on the free Hamiltonian $H_0$ and the potential $V:=H-H_0$ guaranteeing
the validity of Assumption \ref{Ass_wave} will now be presented. These conditions are
natural extensions of the compactness of the support of the potential and the virial
condition appearing in the case $\Phi(q,p)=q$, $H_0(q,p):=|p|^2/2$ and
$H(q,p)=|p|^2/2+V(q)$ on $M=T^*\R^n$. As in Section \ref{Sec_wave}, we always assume
that the Hamiltonians $H_0$ and $H$ have complete flows $\{\varphi^0_t\}_{t\in\R}$
and $\{\varphi_t\}_{t\in\R}$.

We start below with a result on the domain and the range of the wave maps $W_\pm$. For
this, we need to introduce the sets of $\Phi$-bounded trajectories
$$
B_\Phi^\pm:=\big\{m\in M\mid\exists\;\!R\ge0\hbox{ such that }
\big|\Phi\big(\varphi_{\pm t}(m)\big)\big|\le R\hbox{ for all }t\ge0\big\}.
$$
The sets $B_\Phi^\pm$ coincide with the usual sets of bounded trajectories
\cite[Def.~2.1.1]{DG97} if the map $|\Phi|:M\to[0,\infty)$ is proper. If it is
not the case, the inclusion $m\in B_\Phi^-\cap B_\Phi^+$ does not guarantee that the
orbit $\{\varphi_t(m)\}_{t\in\R}$ stays in a compact subset of $M$ (consider for
instance the case $M=T^*\R^2$, $\Phi(q,p)=q_1$ and $H(q,p):=|p|^2/2$).

We also need the following assumption on the potential\;\!:

\begin{Assumption}[Potential]\label{Ass_pot}
The difference $V\equiv H-H_0$ is of bounded support in $\Phi$, \ie there exists a
constant $R_V\ge0$ such that $|\Phi(m)|\le R_V$ for all $m\in\supp(V)$.
\end{Assumption}

In explicit situations, Assumption \ref{Ass_pot} can often be relaxed to some
condition on the decay at infinity (see for instance \cite[Sec.~2.6-2.7]{DG97},
\cite{Her74}, \cite{Hun68}, \cite[Sec.~3]{JO91} and \cite{Sim71}).

\begin{Theorem}[Existence of wave maps]\label{t-thithi}
Let $H_0$ and $H$ satisfy Assumptions \ref{AssCom} and \ref{Ass_pot}. Then $W_\pm$
exist and are symplectomorphisms from $M\setminus\Crit(H_0,\Phi)$ to
$M\setminus B_\Phi^\pm$. Furthermore, one has $H\circ W_{\pm}=H_0$ on
$M\setminus\Crit(H_0,\Phi)$ (and thus $H=H_0\circ W_{\pm}^{-1}$ on $M\setminus
B_\Phi^\pm$).
\end{Theorem}

Theorem \ref{t-thithi} implies Assumptions \ref{Ass_wave}(i)-(ii) with
$\D_\pm=M\setminus\Crit(H_0,\Phi)$ and $\Ran(W_\pm)=M\setminus B_\Phi^\pm$, but it
does not imply Assumption \ref{Ass_wave}(iii) since we do not have the equality
$\Ran(W_+)=\Ran(W_-)$. Theorem \ref{t-thithi} also implies that the sets $B_\Phi^\pm$
are closed in $M$.

\begin{proof}
We give the proof for $W_+$, since the proof for $W_-$ is similar.

(i) Let $K\subset M\setminus\Crit(H_0,\Phi)$ be compact. Then, there exists $T>0$
such that $\big|\Phi\big(\varphi^0_t(m)\big)\big|>R_V+1$ for all $t>T$ and $m\in K$.
In particular, $\varphi^0_t(K)\cap\supp(V)=\varnothing$ for all $t>T$. Thus, one has
$\varphi^0_t=\varphi^0_{t-T}\circ\varphi^0_{T}=\varphi_{t-T}\circ\varphi^0_T$ on $K$
for all $t>T$, which implies
$$
\lim_{t\to\infty}\varphi_{-t}\circ\varphi^0_t
=\lim_{t\to\infty}\varphi_{-t}\circ\varphi_{t-T}\circ\varphi^0_T
=\varphi_{-T}\circ\varphi^0_T
$$
on $K$. Thence, $W_+$ is a hamiltomorphism on any compact subset of
$M\setminus\Crit(H_0,\Phi)$, and thus a symplectomorphism on
$M\setminus\Crit(H_0,\Phi)$.

Let $m\in W_+\big(M\setminus\Crit(H_0,\Phi)\big)$. Then, there exist
$m_0\in M\setminus\Crit(H_0,\Phi)$ and $T>0$ such that 
$m=W_+(m_0)=\big(\varphi_{-t}\circ\varphi^0_t\big)(m_0)$ for all $t>T$, meaning that
$\varphi_t(m)=\varphi^0_t(m_0)$ for all $t>T$. However, since
$m_0\in M\setminus\Crit(H_0,\Phi)$, one has for each $t\in\R$
$$
(\nabla H_0)\big(\varphi^0_t(m_0)\big)=(\nabla H_0)(m_0)\neq0
\quad\hbox{and}\quad
\big|\Phi\big(\varphi^0_t(m_0)\big)\big|
\ge\big||t|\cdot|(\nabla H_0)(m_0)|-|\Phi(m_0)|\big|
$$
due to \eqref{Eq_linear}-\eqref{Eq_linear_bis}. This implies that
$m\in M\setminus B_\Phi^+$.

Assume now that $m \in M\setminus B_\Phi^+$. If
$\big|\Phi\big(\varphi_t(m)\big)\big|\ge R_V+1$ for all $t\ge0$, one directly obtains
$m=W_+(m)\in W_+\big(M\setminus\Crit(H_0,\Phi)\big)$. If not, consider the time 
$t_0:=\inf\big\{t\ge0\mid\big|\Phi\big(\varphi_t(m)\big)\big|>R_V+1\big\}$ and
observe that for each $\varepsilon\in(0,t_0)$
\begin{align*}
(R_V+1)^2
>\big|\Phi\big(\varphi_{t_0-\varepsilon}(m)\big)\big|^2
&=\left|\Phi\big(\varphi_{t_0}(m)\big)
-\int_0^\varepsilon\d s\,\{\Phi,H\}\big(\varphi_{t_0-s}(m)\big)\right|^2\\
&=\big|\Phi\big(\varphi_{t_0}(m)\big)
-\varepsilon\;\!(\nabla H_0)\big(\varphi_{t_0}(m)\big)\big|^2\\
&=(R_V+1)^2-2\varepsilon\;\!(\Phi\cdot\nabla H_0)\big(\varphi_{t_0}(m)\big)
+\varepsilon^2\big|(\nabla H_0)\big(\varphi_{t_0}(m)\big)\big|^2.
\end{align*}
So, one has
$
2(\Phi\cdot\nabla H_0)\big(\varphi_{t_0}(m)\big)
>\varepsilon\big|(\nabla H_0)\big(\varphi_{t_0}(m)\big)\big|^2,
$
which implies $(\nabla H_0)\big(\varphi_{t_0}(m)\big)\neq0$ and
$(\Phi\cdot\nabla H_0)\big(\varphi_{t_0}(m)\big)>0$. It follows that
$$
\frac\d{\d t}\big|\Phi\big(\varphi_{t_0+t}(m)\big)\big|^2\Big|_{t=0}
=\frac\d{\d t}\left\{(R_V+1)^2+2t(\Phi\cdot\nabla H_0)\big(\varphi_{t_0}(m)\big)
+t^2\big|(\nabla H_0)\big(\varphi_{t_0}(m)\big)\big|^2\Big\}\right|_{t=0}
>0.
$$
So, $\varphi_{t_0+t}(m)$ remains out of the support of $V$ for all $t\ge0$, which
implies $\varphi_{t+t_0}(m) = \big(\varphi^0_t \circ \varphi_{t_0}\big) (m)$ for all
$t\ge0$. To conclude, let $m_0\in M\setminus\Crit(H_0,\Phi)$ satisfy
$\varphi^0_{t_0}(m_0)=\varphi_{t_0}(m)$ (such a $m_0$ exists, since
$\varphi_{t_0}(m)\in M\setminus\Crit(H_0,\Phi)$ and $\varphi^0_{t_0}$ is a
diffeomorphism of $M\setminus\Crit(H_0,\Phi)$ onto itself). Then the last formula
gives
$
W_+(m_0)
=\lim_{t\to\infty}\big(\varphi_{-t-t_0}\circ\varphi^0_t\circ\varphi_{t_0}\big)(m)
=m
$,
which implies $m\in W_+\big(M\setminus\Crit(H_0,\Phi)\big)$.

(ii) Take $m\in M\setminus\Crit(H_0,\Phi)$. Then we know from (i) that there exists
$T>0$ such that $\varphi_T^0(m)\notin\supp(V)$ and
$W_+(m)=\big(\varphi_{-T}\circ\varphi^0_T\big)(m)$. This implies that
$\big(H\circ W_+\big)(m)=H_0(m)$. 
\end{proof}

In order to define the scattering map $S\equiv W_+^{-1}\circ W_-$, the ranges of the
wave maps $W_\pm$ have to be equal (see Assumption \ref{Ass_wave}(iii)). We present
in the sequel two methods to prove this equality. The first one hinges on a virial
identity, while the second one consists in showing that the symmetric difference of
the ranges of $W_\pm$ is of (Liouville) measure zero. In the second case,
completeness holds upon removing from $M$ a relatively abstract set; namely, the
preimage of a set of measure zero. In classical mechanics, this type of completeness is sometimes referred as
asymptotic completeness \cite[Sec.~3.4]{Thi97}.

In the standard case, the virial method relies on the following observation: if there
exists $\delta>0$ such that
$$
\frac{\d^2}{\d t^2}\big(|\Phi|^2\circ\varphi_t\big)(m)
\equiv\big(\big\{\big\{|\Phi|^2,H\big\},H\big\}\circ\varphi_t\big)(m)>\delta
\quad\hbox{for all $t\in\R$ and $m\in M$,}
$$
then $\lim_{|t|\to\infty}\big(|\Phi|^2\circ\varphi_t\big)(m)=+\infty$ for all
$m\in M$, and so $B_\Phi^\pm=\varnothing$. Accordingly, it is sufficient to prove
that  $\big\{\big\{|\Phi|^2,H\big\},H\big\}>\delta$ on $M$ to get, under Assumptions
\ref{AssCom} and \ref{Ass_pot}, the completeness of the wave maps. Now, a direct
calculation shows that the expression for $\big\{\big\{|\Phi|^2,H\big\},H\big\}$ (the
virial identity) is
\begin{equation}\label{daunting}
\frac12\;\!\big\{\big\{|\Phi|^2,H\big\},H\big\}
=|\nabla H_0|^2+|\{\Phi,V\}|^2+\Phi\cdot\big\{\{\Phi,V\},V\big\}
+\big\{\Phi\cdot\{\Phi,V\},H_0\big\}+\{\Phi\cdot\nabla H_0,V\}.
\end{equation}
At this level of generality, finding scattering systems $(M,H_0,H)$ for which this
expression is bounded away from zero is rather daunting. However, if one assumes that
$\{\Phi,V\}=0$ (as in the standard case where $V$ depends only on the position
$\Phi(q,p)=q$), then Formula \eqref{daunting} reduces to the much more sympathetic
equation\;\!:
$$
\big\{\Phi\cdot\nabla H_0,H\big\}=|\nabla H_0|^2+\Phi\cdot\big\{\nabla H_0,V\big\}.
$$
Lemma \ref{t-condvir} below provides conditions under which one recovers this
simplified situation. For it, we need the following\;\!:

\begin{Assumption}\label{Ass_H_not}
\begin{enumerate}
\item[(i)] $H_0$ is boundedly preserved by the flow of $H$, \ie for each $m\in
M$ there exists a constant $\textsc c_m\ge0$ such that 
$
\big|H_0\big(\varphi_t(m)\big)-H_0(m)\big|\le\textsc c_m
$
for all $t\in\R$.
\item[(ii)] There exists an increasing function $\rho:[0,\infty)\to[0,\infty)$
with
$\lim_{R\to\infty}\rho(R)=+\infty$ such that $|(\nabla H_0)(m)|\ge R$ implies
$|H_0(m)|\ge\rho(R)$ for all $m\in M$ and $R\ge0$.
\end{enumerate}
\end{Assumption}

If $V$ is bounded, then (i) holds automatically, since
$
\big|H_0\big(\varphi_t(m)\big)-H_0(m)\big|=\big|V\big(\varphi_t(m)
\big)-V(m)\big|.
$
On the other hand, (i) also holds for some unbounded $V$'s. For example, if $H_0\ge0$
and $V$ is bounded from below by $V_0\in\R$, then (i) is verified with
$\textsc c_m=H_0(m)+H(m)-V_0$.

\begin{Lemma}[Completeness of wave maps]\label{t-condvir}
Let $H_0$ and $H$ satisfy Assumptions \ref{AssCom} and \ref{Ass_pot}. Suppose
either
that $\{\Phi,V\}=0$ or that Assumption \ref{Ass_H_not} holds. Assume there
exists
$\delta>0$ such that $\{\Phi\cdot\nabla H_0,H\}(m)>\delta$ for all $m\in M$.
Then, $B_\Phi^\pm=\varnothing$ and the maps
$$
W_\pm:M\setminus\Crit(H_0,\Phi)\to M\quad\hbox{and}\quad
S=W_+^{-1}\circ W_-:M\setminus\Crit(H_0,\Phi)\to M\setminus\Crit(H_0,\Phi)
$$
are well defined symplectomorphisms. In particular, Assumptions
\ref{Ass_wave}(i)-(iii) hold on the submanifold $M\setminus\Crit(H_0,\Phi)$.
\end{Lemma}

\begin{proof}
If $\{\Phi,V\}=0$, then the claim follows directly from the observations made
before
Assumption \ref{Ass_H_not}.

So, suppose that Assumption \ref{Ass_H_not} holds. Since
$\{\Phi\cdot\nabla H_0,H\}>\delta$, we have
$\frac\d{\d t}(\Phi\cdot\nabla H_0)\big(\varphi_t(m)\big)>\delta$ for all $t\in\R$
and $m\in M$. In particular, there exist for all $m\in M$ and all $R\ge0$ times
$t_\pm\in\R_\pm$ such that either $\big|\Phi\big(\varphi_{t_\pm}(m)\big)\big|>R$ or
$\big|(\nabla H_0)\big(\varphi_{t_\pm}(m)\big)\big|>R$. However, we know by
Assumption \ref{Ass_H_not}(ii) that if
$\big|(\nabla H_0)\big(\varphi_{t_\pm}(m)\big)\big|>R$ then
$\big|H_0\big(\varphi_{t_\pm}(m)\big)\big| >\rho(R)$. We also know from Assumption
\ref{Ass_H_not}(i) that there exists $\textsc c_m\ge0$ such that  
$\big|H_0\big(\varphi_{t_\pm}(m)\big)-H_0(m)\big|\le\textsc c_m$. This implies that
$$
\rho(R)
<\big|H_0\big(\varphi_{t_\pm}(m)\big)\big|
=\big|H_0\big(\varphi_{t_\pm}(m)\big)-H_0(m)+H_0(m)\big|
\le\textsc c_m+|H_0(m)|,
$$
which is a contradiction since $\rho(R)$ is not bounded for $R$ big enough. Thus, we
must have the following: for all $m\in M$ and all $R\ge0$ such that
$\rho(R)\ge C_m+H_0(m)$, there exist times $t_\pm\in\R_\pm$ such that 
$\big|\Phi\big(\varphi_{t_\pm}(m)\big)\big|>R$. In particular, we have that 
$\lim_{|t|\to\infty}\big(|\Phi|^2\circ\varphi_t\big)(m)=+\infty$ for all $m\in M$,
which implies the claim.
\end{proof}

Let $U\subset\R$ be an open set such that
$H_0^{-1}(U)\cap\Crit(H_0,\Phi)=\varnothing$. Then, $H_0^{-1}(U)$ is a submanifold of
$M$ preserved by the flow of $H_0$. But in general, $H_0^{-1}(U)$ is not preserved by
the flow of $H$. However, if Theorem \ref{t-thithi} applies, one has
$H\circ W_\pm=H_0$ and $W_\pm\big(H_0^{-1}(U)\big)=H^{-1}(U)\setminus B_\Phi^\pm$ is
also a submanifold of $M$. Therefore, the following corollary is a consequence of
Theorem \ref{t-thithi} and Lemma \ref{t-condvir}.

\begin{Corollary}\label{viriel}
Let $H_0$ and $H$ satisfy Assumptions \ref{AssCom} and \ref{Ass_pot}. Suppose
either
that $\{\Phi,V\}=0$ or that Assumption \ref{Ass_H_not} holds. Let $U\subset\R$
be an open
set such that
\begin{enumerate}
\item[(i)] $H_0^{-1}(U)\cap\Crit(H_0,\Phi)=\varnothing$,
\item[(ii)] there exists $\delta>0$ such that
$\big\{\Phi\cdot\nabla H_0,H\big\}(m)>\delta$ for all $m\in H_0^{-1}(U)$.
\end{enumerate}
Then, the sets $H_0^{-1}(U)$ and $H^{-1}(U)$ are submanifolds of $M$, and the maps
$$
W_\pm:H_0^{-1}(U)\to H^{-1}(U)\quad\hbox{and}\quad
S=W_+^{-1}\circ W_-:H_0^{-1}(U)\to H_0^{-1}(U)
$$
are well defined symplectomorphisms. In particular, Assumptions
\ref{Ass_wave}(i)-(iii) hold on the submanifold $H_0^{-1}(U)$. 
\end{Corollary}

Before pursuing the Examples \ref{Ex_hp_1}-\ref{Ex_Poin_1} of Section
\ref{Sec_Free_Ham}, we give a last result on scattering theory sometimes referred to
as asymptotic completeness. It is is inspired by \cite[Thm.~3.4.7(c)]{Thi97} (in the
case $\Phi(q,p)=q$ and $H(q,p)=|p|^2/2+V(q)$ on $M=T^*\R^n$) and basically states
that the ranges of $W_\pm$ are equal up to a set of Liouville measure zero. We recall
that the Liouville measure of a Borel subset $U\subset M$, with $M$ of finite
dimension, is given by 
$$
\m_\omega(U):=\int_U\frac{\omega^n}{n\;\!!}\;\!.
$$
We also recall that the symmetric difference of sets $X,Y$ is
$X\triangle Y:=(X\setminus Y)\cup(Y\setminus X)$.

\begin{Proposition}[Asymptotic completeness of wave maps]\label{t-compasympt}
Assume that $M$ has finite dimension. Let $H_0$ and $H$ satisfy Assumptions
\ref{AssCom} and \ref{Ass_pot}. Suppose that for each $k\in\N$ the set
$$
A_k:=\big\{m\in M\mid|\Phi(m)|\le k\hbox{ and }\,|H(m)|\le k\big\}
$$
satisfies $\m_\omega(A_k)<\infty$. Then,
$\m_\omega\big(\Ran(W_+)\triangle\Ran(W_-)\big)=0$.
\end{Proposition}

In the case $\Phi(q,p)=q$ and $H(q,p)=|p|^2/2+V(q)$ on $M=T^*\R^n$, the condition
$\m_\omega(A_k)<\infty$, is satisfied, for instance, if $V$ is bounded.

\begin{proof}
For each $k\in\N$, let $A_k^\pm:=\bigcap_{\pm t\ge0}\varphi_t(A_k)$. Then, simple
calculations using the identity $H\circ\varphi_t=H$ show that
$B_\Phi^\pm=\bigcup_{k\in\N}A_k^\pm$. Furthermore, one deduce from Schwarzschild's
capture theorem \cite[Thm.~7.1.15]{AMR88} that
$
\m_\omega(A_k^+\cap A_k^-)=\m_\omega(A_k^+)=\m_\omega(A_k^-)
$
for each $k\in\N$. Therefore, $m_\omega\big(A^\pm_k\setminus A^\mp_k\big)=0$ for each
$k\in\N$, and one gets from Theorem \ref{t-thithi} and the sub-additivity of
$\m_\omega$ that
$$
m_\omega\big(\Ran(W_+)\setminus\Ran(W_-)\big)
=m_\omega\big(B_\Phi^-\setminus B_\Phi^+\big)
\le\sum_{k\in\N}m_\omega\big(A^-_k\setminus B_\Phi^+\big)
\le\sum_{k\in\N}m_\omega\big(A^-_k\setminus A^+_k\big)
=0.
$$
Repeating the argument with $W_+$ and $W_-$ interchanged yields the conclusion.
\end{proof}

\begin{Example}[$H_0(q,p)=h(p)$, continued]\label{Ex_hp_2}
Let $V\in\cinf_{\rm c}(\R^n;\R)$ be a $\cinf$ function with compact support and let
$H\in\cinf(M;\R)$ be the perturbed Hamiltonian given by $H(q,p):=h(p)+V(q)$. To show
the completeness of the corresponding vector field $X_H$, one has to impose some
condition on $h$. So, we assume that
$|(\nabla h)(p)|\le\alpha\e^{\;\!\beta\langle p\rangle}$ for some constants
$\alpha,\beta\ge0$ and all $p\in\R^n$, but we also note that many other cases can be
covered (such as when $h$ is a proper map). Under this assumption, the completeness
of $X_H$ follows from \cite[Prop.~2.1.20]{AM78} with the (proper and $\cinf$)
function $M\ni(q,p)\mapsto\e^{\;\!\beta\langle p\rangle}+\langle q\rangle\in\R$.
Since Assumption \ref{Ass_pot} is satisfied, Theorem \ref{t-thithi} implies that the
wave maps $W_\pm$ exist and are symplectomorphisms from
$\R^n\times\R^n\setminus(\nabla h)^{-1}(\{0\})$ to $M\setminus B_\Phi^\pm$.
Furthermore, the commutation $\{\Phi,V\}=0$ implies that
$$
\big\{\Phi\cdot\nabla H_0,H\big\}(q,p)
=|(\nabla h)(p)|^2-q^\T\cdot(\Hess h)(p)\;\!(\nabla V)(q)
\ge|(\nabla h)(p)|^2
-n\;\!{\textsc c_V}\max_{1\le i,j\le n}\big|(\partial_i\partial_jh)(p)\big|,
$$
where $(\Hess h)(p)$ is the Hessian matrix of $h$ at $p$ and
${\textsc c_V}:=\sup_{q\in\R^n}|q|\cdot|(\nabla V)(q)|$. Therefore, if there exist
continuous functions $g_1,g_2:\Ran(h)\to[0,\infty)$ such that
$$
|(\nabla h)(p)|^2\ge g_1\big(h(p)\big)
\quad\hbox{and}\quad
\max_{1\le i,j\le n}\big|(\partial_i\partial_jh)(p)\big|
\le g_2\big(h(p)\big)
\quad\hbox{for all}~p\in\R^n
$$
(which occurs for instance when $h(p)=|p|^2/2$ or $h(p)=\sqrt{1+p^2}$), then the open
sets
$$
U_\delta:=\big\{x\in\R\mid g_1(x)-n\;\!{\textsc c_V}g_2(x)>\delta\big\},
\quad\delta>0,
$$
satisfy the conditions (i) and (ii) of Corollary \ref{viriel}. Thus, the maps
$$
W_\pm:H_0^{-1}(U_\delta)\to H^{-1}(U_\delta)\quad\hbox{and}\quad
S=W_+^{-1}\circ W_-:H_0^{-1}(U_\delta)\to H_0^{-1}(U_\delta)
$$
are well defined symplectomorphisms, and Assumption \ref{Ass_wave} holds on the
submanifold $H_0^{-1}(U_\delta)$.
\end{Example}

\begin{Example}[Particle in a tube, continued]\label{Ex_tube_2}

Let $V\in\cinf_{\rm c}(\Omega;\R)$ be a $\cinf$ function with compact support and let
$H\in\cinf(M;\R)$ be the perturbed Hamiltonian given by $H(q,p):=H_0(q,p)+V(q)$.
Then, Assumption \ref{Ass_pot} is satisfied and we know from
\cite[Prop.~2.1.20]{AM78} (once more applied with the function \eqref{Func_tube})
that $X_H$ is complete. So, Theorem \ref{t-thithi} implies that the wave maps $W_\pm$
exist and are symplectomorphisms from
$\Omega\times\big(\R\setminus\{0\}\times\R^n\big)$ to $M\setminus B_\Phi^\pm$.
Now, we cannot apply Corollary \ref{viriel} to obtain the identity of the ranges
of $W_\pm$, since the virial identity
$$
\big\{\Phi\cdot\nabla H_0,H\big\}(q,p)=p_1^2-q^1\big(\partial_1V\big)(q)
$$
does not involve observables comparable with the free energy $H_0(q,p)$. Instead, we
set $K_V:=\sup_{q\in\Omega}\big|q^1\big(\partial_1V\big)(q)\big|$ and define the
open set
$$
M_\delta:=\big\{(q,p)\in M\mid|(\nabla H_0)(q,p)|^2>K_V+\delta\big\},\quad\delta>0,
$$
Then, $\varphi_t^0$ is a diffeomorphism on $M_\delta$ for each $t\in\R$ and
$\big|\big\{\Phi\cdot\nabla H_0,H\big\}\big|>\delta$ on $M_\delta$. So, for each
$(q,p)\in M_\delta$ there exists $T>0$ such that
$W_-(q,p)=\big(\varphi_T\circ\varphi_{-T}^0\big)(q,p)$ and
$$
\lim_{|t|\to\infty}\big(|\Phi|^2\circ\varphi_t\circ W_-\big)(q,p)
=\lim_{|t'|\to\infty}\big(|\Phi|^2\circ\varphi_{t'}\circ\varphi_{-T}^0\big)(q,p)
=+\infty.
$$
This implies that $W_-(M_\delta)\subset M\setminus\big\{B^-_\Phi\cup B^+_\Phi\big\}$,
and thus
$$
S:=W_+^{-1}\circ W_-:M_\delta\to\big(W_+^{-1}\circ W_-\big)(M_\delta)
$$
is a well defined symplectomorphism.
\end{Example}

\begin{Example}[Poincar\'e ball, continued]\label{Ex_Poin_2}
Let $V\in\cinf(\mathring B_1;\R)$ with $\supp(V)\subset B_{R_V}$ for some
$R_V\in[0,1)$ and let $H\in\cinf(M;\R)$ be the perturbed Hamiltonian given by
$H(q,p):=H_0(q,p)+V(q)$. Then, Assumption \ref{Ass_pot} is satisfied (since
$|\Phi|\le\tanh^{-1}(2R_V)$ on $\supp(V)\times\R^n$) and we know from Gordon's
Theorem \cite{Gor70} that $X_H$ is complete. So, Theorem \ref{t-thithi} implies that
the wave maps $W_\pm$ exist and are symplectomorphisms from
$\mathring B_1\times\R^n\setminus\{0\}$ to $M\setminus B_\Phi^\pm$. Now, direct
calculations using the inclusion $\supp(V)\subset B_{R_V}$ and the bound
$|\Phi|\le\tanh^{-1}(2R_V)$ on $\supp(V)\times\R^n$ show that
$$
\big\{\Phi\cdot\nabla H_0,H\big\}
=2H_0+\sqrt{2H_0}\,\big\{\Phi,V\big\}+\Phi\;\!\big\{\sqrt{2H_0},V\big\}
$$
with $\big|\sqrt{2H_0}\,\big\{\Phi,V\big\}\big|$ and
$\big|\Phi\;\!\big\{\sqrt{2H_0},V\big\}\big|$ bounded on $M$. Therefore, there exists
a constant $K_V\ge0$ such that $\big\{\Phi\cdot\nabla H_0,H\big\}\ge2H_0-K_V$ on $M$.
Moreover, Assumption \ref{Ass_H_not} is satisfied due to the boundedness of $V$ and
the identity $\nabla H_0=\sqrt{2H_0}$. So, the open intervals
$U_\delta:=\big(\frac{K_V+\delta}2,\infty\big)$, $\delta>0$, satisfy the conditions
(i) and (ii) of Corollary \ref{viriel}, and the maps
$$
W_\pm:H_0^{-1}(U_\delta)\to H^{-1}(U_\delta)\quad\hbox{and}\quad
S=W_+^{-1}\circ W_-:H_0^{-1}(U_\delta)\to H_0^{-1}(U_\delta)
$$
are well defined symplectomorphisms. In particular, Assumption \ref{Ass_wave} holds
on the submanifold $H_0^{-1}(U_\delta)$.
\end{Example}

\section{Time delay in classical scattering theory} \label{Sec_time}
\setcounter{equation}{0}

In this section, we consider for general scattering systems $(M,H_0,H)$ the
symmetrised time delay defined in terms of sojourn times in the dilated regions
$\Phi^{-1}(B_r)$. Under appropriate assumptions, we prove its existence and relate
it to the time of arrival $T$ defined in \eqref{def_T}. When the scattering process
preserves the norm of the velocity observable $\nabla H_0$, we show that the original
(unsymmetrised) time delay also exists and coincides with the symmetrised time delay.
We refer to \cite{BD81,BO81}, \cite[Sec.~II.B]{BGG83}, \cite[Sec.~4.1]{dCN02},
\cite[Sec.~III]{GT07}, \cite[Ch.~10]{KK92}, \cite{Nar80,NT81} and
\cite[Sec.~3.4]{Thi97} for previous works on classical time delay for
$H_0(q,p)=|p|^2/2$ and $H(q,p)=|p|^2/2+V(q)$ on $M=T^*\R^n$.

So, let
$$
T_r^0(m_-):=\int_\R\d t\,\big(\chi_r^\Phi\circ\varphi_t^0\big)(m_-)
$$
be the sojourn time in the region $\Phi^{-1}(B_r)$ of the free trajectory starting
from $m_-\in\D_-$ at time $t=0$, and let
$$
T_r(m_-):=\int_\R\d t\,\big(\chi_r^\Phi\circ\varphi_t\circ W_-\big)(m_-)
$$
be the corresponding sojourn time of the perturbed trajectory starting from
$W_-(m_-)$ at time $t=0$. The free sojourn time $T_r^0(m_-)$ is finite for each
$m_-\in\D_-\setminus\Crit(H_0,\Phi)$ due to Equation \eqref{Eq_linear}. The
finiteness of the perturbed sojourn time $T_r(m_-)$ is shown in Lemma \ref{Lem_soj}
below under some additional assumptions. Under these assumptions, one can define the
symmetrised time delay in $\Phi^{-1}(B_r)$ for the scattering system $(M,H_0,H)$ with
starting point $m_-:$
$$
\tau_r(m_-):=T_r(m_-)-\frac12\;\!\big\{T_r^0(m_-)+\big(T_r^0\circ
S\big)(m_-)\big\}.
$$
The time $\tau_r(m_-)$ can be interpreted as the time spent the perturbed trajectory
$\big\{\big(\varphi_t\circ W_-\big)(m_-)\big\}_{t\in\R}$ within $\Phi^{-1}(B_r)$
minus the time spent by the corresponding free trajectory (before and after
scattering) within the same region.

In the next lemma, we use the auxiliary time
\begin{align}
\tau_r^{\rm free}(m_-)
:=\frac12\int_0^\infty\d t\,&\big\{\big(\chi_r^\Phi\circ\varphi_t^0\circ
S\big)(m_-)
-\big(\chi_r^\Phi\circ\varphi_{-t}^0\circ S\big)(m_-)\label{tau_free}\\
&\qquad\qquad-\big(\chi_r^\Phi\circ\varphi_t^0\big)(m_-)
+\big(\chi_r^\Phi\circ\varphi_{-t}^0\big)(m_-)\big\},\nonumber
\end{align}
which is finite for all $m_-\in\D_-$ due to Theorem \ref{Cor_car}. The definition of $\tau_r^{\rm free}$ is inspired by a similar definition in the context of quantum
scattering theory \cite[Sec.~4]{RT11}.

\begin{Lemma}\label{Lem_soj}
Let $H_0$ and $H$ satisfy Assumptions \ref{AssCom} and \ref{Ass_wave}, and let
$m_-\in\D_-\setminus\Crit(H_0,\Phi)$ satisfy $S(m_-)\notin\Crit(H_0,\Phi)$. Suppose
also that there exist functions $g_\pm\in\lone(\R_\pm,\d t)$ such that
\begin{equation}\label{hipo}
\big|\big(\chi_r^\Phi\circ W_-
-\chi_r^\Phi\big)\big(\varphi_t^0(m_-)\big)\big|
\le g_-(t)\quad\hbox{for all}~\,r>0\hbox{ and }\;\!t\in\R_-
\end{equation}
and
\begin{equation}\label{hipopo}
\big|\big(\chi_r^\Phi\circ W_+
-\chi_r^\Phi\big)\big((S\circ\varphi_t^0)(m_-)\big)\big|
\le g_+(t)\quad\hbox{for all}~\,r>0\hbox{ and }\;\!t\in\R_+.
\end{equation}
Then, $T_r(m_-)$ is finite for each $r>0$, and
$$
\lim_{r\to\infty}\big\{\tau_r(m_-)-\tau_r^{\rm free}(m_-)\big\}=0.
$$
\end{Lemma}

Before moving on to the proof of the lemma, we make a digression to show that the
conditions \eqref{hipo}-\eqref{hipopo} are automatically verified if Assumption
\ref{Ass_pot} holds\;\!:

\begin{Remark}\label{rk_hipo_po}
If Assumptions \ref{AssCom}, \ref{Ass_wave} and \ref{Ass_pot} hold, then we know from
the proof of Theorem \ref{t-thithi} that there exists $T>0$ such that 
$\big(W_+\circ S\big)(m_-)=\big(\varphi_{-T}\circ\varphi^0_T\circ S\big)(m_-)$ and
$\big(S\circ\varphi_t^0\big)(m_-)=\big(\varphi_{t-T}\circ\varphi_T^0\circ S\big)(m_-)$
for all $t>T$. This implies for all $t>T$ that
$$
\big(W_+\circ S\circ\varphi_t^0\big)(m_-)
=\big(\varphi_t\circ W_+\circ S\big)(m_-)
=\big(\varphi_{t-T}\circ\varphi_T^0\circ S\big)(m_-)
=\big(S\circ\varphi_t^0\big)(m_-),
$$
and thus \eqref{hipopo} is satisfied for some $g_+$ of compact support. So, both
\eqref{hipo} and \eqref{hipopo} hold, since a similar argument applies to
\eqref{hipo}.
\end{Remark}

\begin{proof}[Proof of Lemma \ref{Lem_soj}]
Direct computations using the identities
$$
\big(\varphi_t\circ W_-\big)(m_-)
=\big(W_-\circ\varphi_t^0\big)(m_-)
=\big(W_+\circ S\circ\varphi_t^0\big)(m_-)
$$
imply that
\begin{align}
I_r(m_-)
&:=T_r(m_-)-\frac12\;\!\big\{T_r^0(m)+\big(T_r^0\circ S\big)(m_-)\big\}
-\tau_r^{\rm free}(m_-)\label{others}\\
&=\int_{\R_+}\d t\,\big(\chi_r^\Phi\circ
W_+-\chi_r^\Phi\big)\big((S\circ\varphi_t^0)(m_-)\big)
+\int_{\R_-}\d t\,\big(\chi_r^\Phi\circ W_-
-\chi_r^\Phi\big)\big(\varphi_t^0(m_-)\big).\nonumber
\end{align}
It follows by \eqref{hipo} and \eqref{hipopo} that
$$
|I_r(m_-)|\le\int_{\R_+}\d t\,g_+(t)+\int_{\R_-}\d t\,g_-(t)<\infty,
$$
and thus $|I_r(m_-)|$ is bounded by a constant independent of $r$. So, $T_r(m_-)$ is
finite for each $r>0$, since all the other terms of \eqref{others} are finite for
each $r>0$. Moreover, one obtains that $\lim_{r\to\infty}I_r(m_-)=0$ by using
Lebesgue's dominated convergence theorem and the fact that
$\lim_{r\to\infty}\chi_r^\Phi(m)=1$ for each $m\in M$.
\end{proof}

The next theorem shows the existence of the symmetrised time delay $\tau_r(m_-)$ as
$r\to\infty$. It is a direct consequence of Definition \eqref{tau_free}, Lemma
\ref{Lem_soj} and Theorem \ref{Cor_car}.

\begin{Theorem}[Symmetrised time delay]\label{Teo_sym}
Let $H_0$ and $H$ satisfy Assumptions \ref{AssCom} and \ref{Ass_wave}, and let
$m_-\in\D_-\setminus\Crit(H_0,\Phi)$ satisfy $S(m_-)\notin\Crit(H_0,\Phi)$ and
\eqref{hipo}-\eqref{hipopo}. Then, one has 
\begin{equation}\label{Form_sym}
\lim_{r\to\infty}\tau_r(m_-)=T(m_-)-(T\circ S)(m_-).
\end{equation}
\end{Theorem}

Taking into account the definition \eqref{def_T} of $T$, one can rewrite
\eqref{Form_sym} as
$$
\lim_{r\to\infty}\tau_r(m_-)
=\frac{\Phi(m_-)\cdot(\nabla H_0)(m_-)}{|(\nabla H_0)(m_-)|^2}
-\frac{(\Phi\circ S)(m_-)\cdot(\nabla H_0\circ S)(m_-)}
{|(\nabla H_0\circ S)(m_-)|^2}\;\!.
$$

\begin{Remark}\label{Rem_EW}
It is worth making a couple of observations on the result of Theorem
\ref{Teo_sym}\;\!:

\begin{enumerate}
\item[(i)] For fixed $r>0$, the l.h.s. of Formula \eqref{Form_sym} is equal to the
symmetrised time delay in $\Phi^{-1}(B_r)$ for the scattering system $(M,H_0,H)$
with starting point $m_-$. On the other hand, the r.h.s. of Formula \eqref{Form_sym} is
equal to the arrival time $-(T\circ S)(m_-)$ of the particle after scattering minus
the arrival time $-T(m_-)$ of the particle before scattering. Therefore, Formula
\eqref{Form_sym} shows in a very general set-up that this difference of arrival
times is equal to the limit of the symmetrised time delay in $\Phi^{-1}(B_r)$ as
$r\to\infty$.

\item[(ii)] Denote by $\tau(m_-):=\lim_{r\to\infty}\tau_r(m_-)$ the global
time delay obtained in Theorem \ref{Teo_sym}. Then, the linear evolution
$T\circ\varphi^0_t=T+t$ of $T$ under the free flow $\varphi^0_t$, together with
the commutation \eqref{commutation_S} of $S$ with $\varphi^0_t$, implies that
$$
\tau(m_-)
=\big\{(T-T\circ S)\circ\varphi^0_t\big\}(m_-)
=\big(\tau\circ\varphi^0_t\big)(m_-)
$$
for all $t\in\R$, meaning that $\tau$ is a first integral of the free motion. This
property corresponds in the quantum case to the fact that the time delay operator is
decomposable in the spectral representation of the free Hamiltonian (see
\cite[Rem.~4.4]{RT11}).

\item[(iii)] Formula \eqref{Form_sym} can be considered as classical version of
the Eisenbud-Wigner formula of quantum mechanics. Indeed, if one replaces $m_-$ by
an appropriate incoming state $\varphi$ in a Hilbert space $\H$, $(H_0,H)$ by a pair of
self-adjoint operators in $\H$, $T$ by a time operator (acting as the differential
operator $i\frac\d{\d H_0}$ in the spectral representation of $H_0$) and $S$ by the
unitary scattering operator for $(H_0,H)$, one recovers the general Eisenbud-Wigner
formula established in Theorem 4.3 of \cite{RT11}\;\!:
$$
\lim_{r\to\infty}\tau_r(\varphi)
=\langle\varphi,T\varphi\rangle_\H-\langle S\varphi,TS\varphi\rangle_\H
=-\langle\varphi,S^*[T,S]\varphi\rangle_\H
=-\left\langle\varphi,iS^*\frac{\d S}{\d H_0}\;\!\varphi\right\rangle_\H.
$$
\end{enumerate}
\end{Remark}

In the next corollary, we show that the unsymmetrised time delay
$$
\tau_r^{\rm in}(m_-):=T_r(m_-)-T_r^0(m_-)
$$
exists and is equal to the symmetrised time delay in the limit $r\to\infty$ if the
scattering process preserves the norm of the velocity vector $\nabla H_0$ (the
superscript ``\hspace{1pt}in'', borrowed from \cite[Sec.~4]{RT11}, refers to
``\hspace{1pt}incoming'' time delay). This result is the classical analogue of
Theorem 5.4 of \cite{RT11} in quantum scattering theory. In the proof, we use the notations
$\cos(x,y):=\frac{x\cdot y}{|x||y|}$ and $\sin(x,y)^2:=1-\cos(x,y)^2$ for vectors
$x,y\in\R^d$.

\begin{Corollary}[Unsymmetrised time delay]\label{Cor_unsym}
Let $H_0$ and $H$ satisfy Assumptions \ref{AssCom} and \ref{Ass_wave}, and let
$m_-\in\D_-\setminus\Crit(H_0,\Phi)$ satisfy \eqref{hipo}-\eqref{hipopo}.
Suppose
also that
\begin{equation}\label{Ass_velo}
|(\nabla H_0)(m_-)|^2=|(\nabla H_0\circ S)(m_-)|^2.
\end{equation}
Then, one has
\begin{equation}\label{Form_not_sym}
\lim_{r\to\infty}\tau_r^{\rm in}(m_-)
=\lim_{r\to\infty}\tau_r(m_-)
=T(m_-)-(T\circ S)(m_-).
\end{equation}
\end{Corollary}

Note that the assumption $S(m_-)\notin\Crit(H_0,\Phi)$ of Lemma \ref{Lem_soj} and
Theorem \ref{Teo_sym} is, here, automatically verified for each
$m_-\in\D_-\setminus\Crit(H_0,\Phi)$ due to the hypothesis \eqref{Ass_velo}.

\begin{proof}
The identity
$$
\tau_r^{\rm in}(m_-)
=\tau_r(m_-)+\frac12\;\!\big\{(T_r^0\circ S)(m_-)-T_r^0(m_-)\big\},
$$
together with Theorem \ref{Teo_sym}, implies that is enough to show that
$$
\lim_{\nu\searrow0}\big\{\big(T_{1/\nu}^0\circ
S\big)(m_-)-T_{1/\nu}^0(m_-)\big\}=0.
$$
Now, we know from the proof of \cite[Lemma~2.4]{GT11} that for any $x\in\R^d$
and $y\in\R^d\setminus\{0\}$ one has
$$
\int_{\R_+}\d t\,\chi_{1/\nu}(x\pm ty)
=\frac{\sqrt{1-\nu^2|x|^2\sin(x,y)^2}}{|y|\nu}
\mp\frac{x\cdot y}{|y|^2}
$$
if $\nu>0$ is small enough. So, a direct calculation using Formula \eqref{Eq_linear}
and the hypothesis \eqref{Ass_velo} gives
\begin{align*}
&\big(T_{1/\nu}^0\circ S\big)(m_-)-T_{1/\nu}^0(m_-)\\
&=\frac{2\sqrt{1-\nu^2|(\Phi\circ S)(m_-)|^2
\sin\big((\Phi\circ S)(m_-),(\nabla H_0\circ S)(m_-)\big)^2}}{|(\nabla H_0\circ
S)(m_-)|\;\!\nu}\\
&\qquad-\frac{2\sqrt{1-\nu^2|\Phi(m_-)|^2\sin\big(\Phi(m_-),(\nabla
H_0)(m_-)\big)^2}}
{|(\nabla H_0)(m_-)|\;\!\nu}\\
&=\frac2{|(\nabla H_0\circ S)(m_-)|\;\!\nu}
\Big\{\sqrt{1-\nu^2|(\Phi\circ S)(m_-)|^2
\sin\big((\Phi\circ S)(m_-),(\nabla H_0\circ S)(m_-)\big)^2}-1\Big\}\\
&\qquad-\frac2{|(\nabla H_0)(m_-)|\;\!\nu}\Big\{\sqrt{1-\nu^2|\Phi(m_-)|^2
\sin\big(\Phi(m_-),(\nabla H_0)(m_-)\big)^2}-1\Big\},
\end{align*}
which implies that
\begin{align*}
&\lim_{\nu\searrow0}\big\{\big(T_{1/\nu}^0\circ
S\big)(m_-)-T_{1/\nu}^0(m_-)\big\}\\
&=\frac2{|(\nabla H_0\circ S)(m_-)|}\;\!\frac\d{\d\nu}\sqrt{1-\nu^2|(\Phi\circ
S)(m_-)|^2
\sin\big((\Phi\circ S)(m_-),(\nabla H_0\circ S)(m_-)\big)^2}\,\bigg|_{\nu=0}\\
&\qquad-\frac2{|(\nabla H_0)(m_-)|}\;\!\frac\d{\d\nu}\sqrt{1-\nu^2|\Phi(m_-)|^2
\sin\big(\Phi(m_-),(\nabla H_0)(m_-)\big)^2}\,\bigg|_{\nu=0}\\
&=0-0.
\end{align*}
\end{proof}

Taking into account the definition \eqref{def_T} of $T$ and the hypothesis
\eqref{Ass_velo}, one can rewrite \eqref{Form_not_sym} as
$$
\lim_{r\to\infty}\tau_r^{\rm in}(m_-)
=\lim_{r\to\infty}\tau_r(m_-)
=\frac{\Phi(m_-)\cdot(\nabla H_0)(m_-)
-(\Phi\circ S)(m_-)\cdot(\nabla H_0\circ S)(m_-)}{|(\nabla H_0)(m_-)|^2}\;\!.
$$

\begin{Remark}\label{rem_sim_pas_sim}
In general, one cannot expect the existence of the unsymmetrised time delay as
$r\to\infty$, since the sojourn times in regions defined in terms of $\Phi$ have no
reason to be comparable before and after scattering (even though $S$ commutes with
$\varphi^0_t$\;\!!). This occurs only in particular situations, as when the
scattering process preserves the norm of the velocity vector $\nabla H_0$. This is in
fact exactly what tells us Corollary \ref{Cor_unsym}: if the scattering process
preserves $|\nabla H_0|$, then the unsymmetrised time delay also exists and is equal
to the symmetrised time delay in the limit $r\to\infty$. We refer to the Examples
\ref{Ex_hp_3} and \ref{Ex_tube_3} below for an illustration of this observation.
\end{Remark}

\begin{Example}[$H_0(q,p)=h(p)$, continued]\label{Ex_hp_3}
We know that Assumptions \ref{AssCom}, \ref{Ass_wave} and \ref{Ass_pot} hold on
the manifold $H_0^{-1}(U_\delta)$, with $S:H_0^{-1}(U_\delta)\to H_0^{-1}(U_\delta)$
(see Example \ref{Ex_hp_2}). It follows that each $(q_-,p^-)\in H_0^{-1}(U_\delta)$ satisfy $S(q_-,p^-)\notin\Crit(H_0,\Phi)$ and \eqref{hipo}-\eqref{hipopo} (see Remark \ref{rk_hipo_po}).
So, Theorem \ref{Teo_sym} applies, and the global time delay exists and satisfies
$$
\lim_{r\to\infty}\tau_r(q_-,p^-)
=T(q_-,p^-)-T(q_+,p^+)
=\frac{q_-\cdot(\nabla h)(p^-)}{|(\nabla h)(p^-)|^2}
-\frac{q_+\cdot(\nabla h)(p^+)}{|(\nabla h)(p^+)|^2}\;\!,
$$
where $(q_+,p^+):=S(q_-,p^-)$. Now, if there exists a diffeomorphism $h_0:(0,\infty)\to\Ran(h_0)$
such that $h(p)=h_0(p^2)$ for all $p\in\R^n$ (such as when $h(p)=p^2/2$), then
$|(\nabla h)(p)|^2=\big(f\circ h\big)(p)$ with $f\in\cinf\big(\Ran(h_0)\big)$ given by
$f(x):=4h_0^{-1}(x)\big|h_0'\big(h_0^{-1}(x)\big)\big|^2$. Therefore, one has for any
$(q_-,p^-)\in H_0^{-1}(U_\delta)$
$$
|(\nabla H_0)(q_-,p^-)|^2
=\big(f\circ H_0\big)(q_-,p^-)
=\big(f\circ H_0\circ S\big)(q_-,p^-)
=|(\nabla H_0\circ S)(q_-,p^-)|^2
$$
due to the identities $H\circ W_{\pm}=H_0$ and $H=H_0\circ W_{\pm}^{-1}$ of Theorem
\ref{t-thithi}. So, Corollary \ref{Cor_unsym} applies, and the unsymmetrised time delay
exists and satisfies
$$
\lim_{r\to\infty}\tau_r^{\rm in}(q_-,p^-)
=\lim_{r\to\infty}\tau_r(q_-,p^-)
=\frac{q_-\cdot(\nabla h)(p^-)-q_+\cdot(\nabla h)(p^+)}{|(\nabla h)(p^-)|^2}\;\!.
$$
\end{Example}

\begin{Example}[Particle in a tube, the end]\label{Ex_tube_3}
We know that Assumptions \ref{AssCom} and \ref{Ass_pot} hold on $M_\delta$ and that
$S:M_\delta\to\big(W_+^{-1}\circ W_-\big)(M_\delta)$ has the required properties
(see Example \ref{Ex_tube_2}). It follows that each $(q_-,p^-)\in M_\delta$ satisfy $S(q_-,p^-)\notin\Crit(H_0,\Phi)$ and \eqref{hipo}-\eqref{hipopo} (see Remark
\ref{rk_hipo_po}). So, Theorem \ref{Teo_sym} applies and the global time delay in the
tube exists and satisfies
$$
\lim_{r\to\infty}\tau_r(q_-,p^-)
=T(q_-,p^-)-T(q_+,p^+)
=\frac{q_-^1p^-_1}{(p^-_1)^2}-\frac{q_+^1p^+_1}{(p^+_1)^2}\;\!,
$$
where $(q_+,p^+):=S(q_-,p^-)$. Note that although the scattering process preserves
the free energy $H_0$, it does not preserve the norm of the longitudinal momentum
alone since rearrangements between the transverse and longitudinal momenta occur
during the scattering. So, we do not have $(p^-_1)^2\neq(p^+_1)^2$ in general, and
thus (in agreement with Corollary \ref{Cor_unsym}) the unsymmetrised time delay has
no reason to exist.
\end{Example}

\begin{Example}[Poincar\'e ball, continued]\label{Ex_Poin_3}
We know that Assumptions \ref{AssCom}, \ref{Ass_wave} and \ref{Ass_pot} hold on
the manifold $H_0^{-1}(U_\delta)$, with $S:H_0^{-1}(U_\delta)\to H_0^{-1}(U_\delta)$
(see Example \ref{Ex_Poin_2}). Furthermore, one has for each $(q_-,p^-)\in H_0^{-1}(U_\delta)$
$$
|(\nabla H_0)(q_-,p^-)|^2
=2H_0(q_-,p^-)
=2(H_0\circ S)(q_-,p^-)
=|(\nabla H_0\circ S)(q_-,p^-)|^2.
$$
So, Corollary \ref{Cor_unsym} applies, and both time delays exist and satisfy
\begin{align*}
\lim_{r\to\infty}\tau_r^{\rm in}(q_-,p^-)
=\lim_{r\to\infty}\tau_r(q_-,p^-)
&=T(q_-,p^-)-T(q_+,p^+)\\
&=\frac{\Phi(q_-,p^-)-\Phi(q_+,p^+)}{\sqrt{2H_0(q_-,p^-)}}\\
&=\frac{2\left\{\tanh^{-1}\left(\frac{2(p^-\cdot\;\!q_-)}{|p^-|(1+|q_-|^2)}\right)
-\tanh^{-1}\left(\frac{2(p^+\cdot\;\!q_+)}{|p^+|(1+|q_+|^2)}\right)\right\}}
{|p^-|(1-|q_-|^2)}\;\!,
\end{align*}
with $(q_+,p^+):=S(q_-,p^-)$.
\end{Example}

\section{Calabi invariant of the Poincar\'e scattering map}\label{Sec_Cal}
\setcounter{equation}{0}

In this section, we relate the Calabi morphism (evaluated at the Poincar\'e scattering
map) to the time delay by combining the results of Section \ref{Sec_time} with the
ones of \cite{BP09}. As in the previous sections, we always assume that the
Hamiltonians $H_0$ and $H$ have complete flows $\{\varphi^0_t\}_{t\in\R}$ and
$\{\varphi_t\}_{t\in\R}$.

So, let $E\in\R$ be such that $H_0^{-1}(\{E\})\cap\Crit(H_0,\Phi)=\varnothing$. Then,  
$\Sigma^0_E:=H_0^{-1}(\{E\})$ is a regular submanifold of $M$ of dimension $2n-1$,
and the map 
$$
\Psi_E:\Sigma^0_E\to\R,\quad m\mapsto(\Phi\cdot\nabla H_0)(m)
$$
is $\cinf$. Furthermore, for each $\alpha\in\R$, the set 
$\Gamma_{E,\alpha}:=\Psi_E^{-1}(\{\alpha\})\subset\Sigma^0_E$ satisfies the
following\;\!:

\begin{Lemma}[Transversal section]\label{poinsec}
Let $H_0$ satisfy Assumption \ref{AssCom} and take $E\in\R$ such that 
$H_0^{-1}(\{E\})\cap\Crit(H_0,\Phi)=\varnothing$. Then, for each $\alpha\in\R$, the
set $\Gamma_{E,\alpha}$ is a regular submanifold of $\Sigma^0_E$ of dimension
$2(n-1)$ such that
\begin{enumerate}
\item[(a)] $X_{H_0}(m)\notin T_m\Gamma_{E,\alpha}$ for all $m\in\Gamma_{E,\alpha}$,
\item[(b)] for all $m\in\Sigma^0_E$, there exists a unique
$m_0=m_0(m)\in\Gamma_{E,\alpha}$ and a unique $t_0=t_0(m)\in\R$ such that
$m=\varphi^0_{t_0}(m_0)$.
\end{enumerate}
\end{Lemma}

Note that the first two assertions imply that $\Gamma_{E,\alpha}$ is (in
$\Sigma^0_E$) a local transversal section of the vector field $X_{H_0}|_{\Sigma^0_E}$
(see \cite[Def.~7.1.1]{AM78}).

\begin{proof}
Take $m\in\Gamma_{E,\alpha}$ and let $\gamma:\R\to\Sigma^0_E$ be the integral curve
of $X_{H_0}$ at $m$ given by $\gamma(t):=\varphi^0_t(m)$. Then, we have
$
(\Psi_E\circ\gamma)(t)=\Psi_E(m)+t\;\!|(\nabla H_0)(m)|^2
$
due to Assumption \ref{AssCom} and Equation \eqref{Eq_linear}. So, the differential
$(\d\Psi_E)_m:T_m\Sigma^0_E\to T_{\Psi_E(m)}\R$ satisfies for each germ
$f\in \cinf_{\Psi_E(m)}(\R)$ at $\Psi_E(m)$ the equalities
\begin{align*}
\big[(\d\Psi_E)_m\big(X_{H_0}(m)\big)\big](f)
=\frac\d{\d t}\;\!(f\circ\Psi_E)(\gamma(t))|_{t=0}
&=\frac\d{\d t}\;\!f\big(\Psi_E(m)+t\;\!|(\nabla H_0)(m)|^2\big)|_{t=0}\\
&=|(\nabla H_0)(m)|^2f'\big(\Psi_E(m)\big)\\
&=|(\nabla H_0)(m)|^2\frac\partial{\partial t}\Big|_{\Psi_E(m)}(f),
\end{align*}
and thus
$
(\d\Psi_E)_m\big(X_{H_0}(m)\big)
=|(\nabla H_0)(m)|^2\frac\partial{\partial t}\big|_{\Psi_E(m)}
$.
Since $|(\nabla H_0)(m)|\neq0$, this implies that $(\d\Psi_E)_m$ is surjective, and
so $\Gamma_{E,\alpha}\equiv\Psi_E^{-1}(\{\alpha\})$ is a regular submanifold of
$\Sigma^0_E$ of codimension $1$ by the regular level set theorem. Moreover, we also
obtain that $X_{H_0}(m)\notin \ker\big((\d\Psi_E)_m\big)$, which implies that
$X_{H_0}(m)\notin T_m\Gamma_{E,\alpha}$ since
$\ker\big((\d\Psi_E)_m\big)=T_m\Gamma_{E,\alpha}$ (see the remark after
\cite[Prop.~1.6.18]{AM78}).

To prove (b), take $m\in\Sigma^0_E$ and observe that
\begin{align*}
\varphi^0_t(m)\in\Gamma_{E,\alpha}
\iff\Psi_E\big(\varphi^0_t(m)\big)=\alpha
\iff\Psi_E(m)+t\;\!|(\nabla H_0)(m)|^2=\alpha
\iff t=\frac{\alpha-\Psi_E(m)}{|(\nabla H_0)(m)|^2}\;\!.
\end{align*}
Thus, the time $t_0:=\frac{\Psi_E(m)-\alpha}{|(\nabla H_0)(m)|^2}\in\R$ and the point
$m_0:=\varphi^0_{-t_0}(m)\in\Gamma_{E,\alpha}$ are unique and satisfy
$m=\varphi^0_{t_0}(m_0)$.
\end{proof}

Lemma \ref{poinsec} implies in particular that the submanifold
$$
\Gamma_E:=\Gamma_{E,0}\equiv\big\{m\in\Sigma^0_E\mid(\Phi\cdot\nabla H_0)(m)=0\big\}
$$
is a Poincar\'e section in the sense of \cite[Assumption~2.2]{BP09}. The rest of the
assumptions of \cite{BP09} are verified in the following lemma\;\!:

\begin{Lemma}\label{Hyp_BP}
Assume that $M$ is exact (that is, with $\omega$ exact) and satisfies
$\dim(M)\ge4$. Suppose also that
\begin{enumerate}
\item[(H1)] Assumption \ref{AssCom} holds,
\item[(H2)] $V$ has compact support,
\item[(H3)] $\{\Phi,V\}=0$ or Assumption \ref{Ass_H_not} holds,
\item[(H4)] $U\subset\R$ is an open set such that
\begin{enumerate}
\item[(i)] $H_0^{-1}(U)\cap\Crit(H_0,\Phi)=\varnothing$,
\item[(ii)] there exists $\delta>0$ such that
$\big\{\Phi\cdot\nabla H_0,H\big\}(m)>\delta$ for all $m\in H_0^{-1}(U)$.
\end{enumerate}
\end{enumerate}
Then, all the assumptions of Theorems 3.1 and 3.2 of \cite{BP09} are verified for
each $E\in U$.
\end{Lemma}

Note that the exactness of $M$ necessarily implies the noncompactness of $M$
\cite[Rem.~V.9.4]{LM87}. Note also that Assumption \ref{Ass_pot} follows from the compactness
of the support of $V$.

\begin{proof}
The hypotheses (H1)-(H4) imply that Corollary \ref{viriel} applies. Thus, each
$E\in U$ is a regular value of $H_0$ and $H$ (Assumption 2.1(i) of \cite{BP09}) and
$W_+\big(H_0^{-1}(\{E\})\big)=W_-\big(H_0^{-1}(\{E\})\big)=H^{-1}(\{E\})$ (Equation
(2.4) of \cite{BP09}). The flows of $H_0$ and $H$ are complete (Assumption 2.1(ii) of
\cite{BP09}). The fact that $H_0^{-1}(U)\cap\Crit(H_0,\Phi)=\varnothing$ implies for
each $E\in U$ the non-trapping condition of Assumption 2.1(iii) of \cite{BP09}; that
is, for any compact set $K\subset\Sigma_E^0$ there exists $T>0$ such that for all
$m\in K$ and all $|t|\ge T$, one has $\varphi_t^0(m)\notin K$. Finally, the sets
$H^{-1}\big((-\infty,E]\big)\cap\supp(V)$ and
$H_0^{-1}\big((-\infty,E]\big)\cap\supp(V)$ are compact for any $E\in\R$ (Assumption
2.1(iv) of \cite{BP09}) due to Assumption (H2).
\end{proof}

\begin{Remark}
In the proof of Lemma \ref{Hyp_BP} we did not check the assumption of noncompactness
of $\widetilde\Sigma^0_E$ made in \cite[Thm.~3.1]{BP09} because we believe it is
unnecessary. Indeed, under the other assumptions of \cite[Thm.~3.1]{BP09}, the authors
of \cite{BP09} show in their Lemma 5.1 that $\widetilde\Sigma^0_E$ is exact.
Therefore, $\widetilde\Sigma^0_E$ is necessarily noncompact, since any exact
symplectic manifold is noncompact (see \cite[Rem.~V.9.4]{LM87}).
\end{Remark}

Under the assumptions of Lemma \ref{Hyp_BP}, we know from \cite{BP09} that the results
of the last paragraph of Section \ref{Sec_wave} hold for any $E\in U :$ The orbit space
$\widetilde\Sigma^0_E=\Sigma^0_E/\R$ is a symplectic manifold of dimension $2(n-1)$
with symplectic form $\widetilde\omega^0_E$ and the restricted scattering map
$S_E:=S|_{\Sigma^0_E}$ induces a symplectomorphism $\widetilde S_E$ of
$\big(\widetilde\Sigma^0_E,\widetilde\omega^0_E\big)$. Furthermore, the Poincar\'e
section $\Gamma_E$ can be considered as a ``concrete realisation of the abstract manifold
$\widetilde\Sigma^0_E$'', due to the existence of a diffeomorphism
$\gamma_E^0:\Gamma_E\to\widetilde\Sigma^0_E$ satisfying
$(\gamma_E^0)^*\widetilde\omega^0_E=\omega|_{\Gamma_E}$ ($\gamma_E^0=\pi_E^0\circ i$, with
$i:\Gamma_E\to\Sigma^0_E$ the natural embedding). Since each element
$m\in\Sigma_E^0$ can be identified with a pair $(m_0,t_0)\in\Gamma_E\times\R$ satisfying
$\varphi_{t_0}^0(m_0)$ due to Lemma \ref{poinsec}(b), we obtain an identification of
$\Sigma^0_E\simeq\Gamma_E\times\R$ which permits to represent the free flow as
$\varphi_t^0:(m_0,t_0)\mapsto(m_0,t_0+t)$ and the map $S_E$ as
$$
S_E:(m_0,t_0)\mapsto\big(\widetilde s_E(m_0),t_0-\tau_E(m_0)\big),
$$
where $\widetilde s_E:=(\gamma_E^0)^{-1}\circ\widetilde S_E\circ\gamma_E^0:\Gamma_E\to\Gamma_E$
is a symplectomorphism and $\tau_E\in\cinf(\Gamma_E;\R)$. Using the expressions for
$(m_0,t_0)$ obtained in the proof of Lemma \ref{poinsec}(b) we thus obtain that
$$
\varphi_{(T\circ S_E)(m)}^0\big(\widetilde s_E(m_0),t_0-\tau_E(m_0)-(T\circ S_E)(m)\big)
=S_E(m)
=\varphi_{-\tau_E(m_0)}^0\big(\widetilde s_E(m_0),t_0\big),
$$
meaning that $\tau_E(m_0)=-(T\circ S_E)(m_0)$. Since $T(m_0)=0$ for each $m_0\in\Gamma_E$, it
follows from Theorem \ref{Teo_sym} that
\begin{equation}\label{Eq_nous_eux}
\lim_{r\to\infty}\tau_r(m_0)=0-(T\circ S_E)(m_0)=\tau_E(m_0)
\quad\hbox{for all $m_0\in\Gamma_E$.}
\end{equation}
This means that, when evaluated at points $m_0\in\Gamma_E$, the global time
delay $\lim_{r\to\infty}\tau_r(m_0)$ defined in terms of sojourn times in
$\Phi^{-1}(B_r)\subset M$ coincides with the time delay $\tau_E(m_0)$ defined on
$\Sigma_E^0$ as the difference of time intervals from the Poincar\'e section
$\Gamma_E$ before and after scattering. In other terms, the choice of the position
observables $\Phi$ provides natural Poincar\'e sections $\Gamma_E$ suitable for the
application of the (fixed energy) theory of \cite{BP09}.

Now, we introduce as in \cite[Sec.~2.5 \& 2.6]{BP09} the average time delay on $\Gamma_E$
$$
\mathcal T_E
:=\int_{\Gamma_E}\tau_E(m_0)\;\!\frac{\omega^{n-1}(m_0)}{(n-1)\;\!!}
=-\int_{\Gamma_E}\frac{(\Phi\circ S)(m_0)\cdot(\nabla H_0\circ S)(m_0)}
{|(\nabla H_0\circ S)(m_0)|^2}\;\!\frac{\omega^{n-1}(m_0)}{(n-1)\;\!!}
$$
and the regularised phase space volume\footnote{Here, we follow the conventions of
\cite[Sec.~2.3]{BP09} for the integrals on $M$ and $\Sigma^0_E:$ The orientation on
$M$ is fixed in such a way that the form $\omega^{n}$ is positive on a positively
oriented basis, and the orientation on $\Sigma^0_E$ is fixed such that if
$(e_1,\ldots,e_{2n-1})$ is a positively oriented basis in $T_m\Sigma^0_E$ and
$v\in T_mM$ is such that $(\d H_0)_m(v)>0$, then $(v,e_1,\ldots,e_{2n-1})$ is a
positively oriented basis in $T_mM$.}
$$
\xi(E):=\int_M\Big(\chi_{H_0^{-1}((-\infty,E])}(m)-\chi_{H^{-1}((-\infty,E])}(m)\Big)
\frac{\omega^{n}(m)}{n\;\!!}\;\!.
$$
Then, the nice Theorems 3.1 and 3.2 of \cite{BP09} state that
\begin{equation}\label{Eq_main_BP}
\mathrm{Cal}\big(\widetilde S_E\big)=\xi(E)
\qquad\hbox{and}\qquad\mathcal T_E=-\frac\d{\d E}\;\!\xi(E),
\end{equation}
where $\mathrm{Cal}:\mathrm{Dom}(\mathrm{Cal},M)\to\R$ is the Calabi invariant as
defined in \cite[Sec.~B.2]{BP09} (the difference with respect to the usual definition
is that $\mathrm{Dom}(\mathrm{Cal},M)$ is here a subset of compactly supported
symplectomorphisms whereas it is usually the set of compactly supported
hamiltomorphisms, see \cite[Eq.~(10.4)]{MS98} or \cite{Ban78}).

By combining Equations \eqref{Eq_nous_eux} and \eqref{Eq_main_BP} one gets the
following\;\!:

\begin{Theorem}[Calabi invariant]\label{Thm_Cal}
Under the assumptions of Lemma \ref{Hyp_BP}, one has for each $E\in U$
$$
\frac\d{\d E}\;\!\mathrm{Cal}\big(\widetilde S_E\big)
=-\int_{\Gamma_E}\lim_{r\to\infty}\tau_r(m_0)\;\!\frac{\omega^{n-1}(m_0)}{(n-1)\;\!!}
=\int_{\Gamma_E}\frac{(\Phi\circ S)(m_0)\cdot(\nabla H_0\circ S)(m_0)}
{|(\nabla H_0\circ S)(m_0)|^2}\;\!\frac{\omega^{n-1}(m_0)}{(n-1)\;\!!}\;\!.
$$
\end{Theorem}

Theorem \ref{Thm_Cal} implies that the derivative of the Calabi invariant evaluated
at $\widetilde S_E$ is equal to the average of the global time delay
$\lim_{r\to\infty}\tau_r$ on $\Gamma_E$ (or equivalently, the average of the arrival
time $T\circ S$ on $\Gamma_E$). Accordingly, it provides a simple and
explicit expression for $\frac\d{\d E}\;\!\mathrm{Cal}\big(\widetilde S_E\big)$ in terms
of $\Phi$, $S$ and $\nabla H_0$ on $\Gamma_E$. Note that Theorem
\ref{Thm_Cal} also holds with $\lim_{r\to\infty}\tau_r$ replaced by
$\lim_{r\to\infty}\tau_r^{\rm in}$ if $|\nabla H_0|^2=|\nabla H_0\circ S|^2$ on
$\Gamma_E$ (see Corollary \ref{Cor_unsym}).

\begin{Example}[$H_0(q,p)=h(p)$, the end]\label{Ex_hp_4}
If the dimension of $M\simeq\R^{2n}$ is bigger or equal to $4$ and $V$ has compact
support, then we know from Example \ref{Ex_hp_3} that all the assumptions of Lemma
\ref{Hyp_BP} are verified on the open set $U_\delta\subset\R$. So, Theorem
\ref{Thm_Cal} applies, and one has for each $E\in U_\delta$
$$
\frac\d{\d E}\;\!\mathrm{Cal}\big(\widetilde S_E\big)
=\int_{\{(q,p)\in\R^{2n}\mid h(p)=E,~q\cdot(\nabla h)(p)=0\}}
\frac{q_+\cdot(\nabla h)(p^+)}{|(\nabla h)(p^+)|^2}
\;\!\frac{\omega^{n-1}(q,p)}{(n-1)\;\!!}\;\!,
$$
with $(q_+,p^+):=S(q,p)$. In the particular case $h(p)=|p|^2/2$, one thus obtains
$$
\frac\d{\d E}\;\!\mathrm{Cal}\big(\widetilde S_E\big)
=(2E)^{(n-3)/2}\int_{\mathbb S^{n-1}}\d^{n-1}\widehat p
\int_{q\cdot p=0}\d^{n-1}q\,\big(q_+\cdot p^+\big),
$$
where $\widehat p:=p/|p|$ and $\d^{n-1}\widehat p$ is the spherical measure on
$\mathbb S^{n-1}$. This corresponds to the case treated in \cite[Sec.~4.2]{BP09}
(when the parameter $R\in\R$ of \cite[Eq.~(4.5)]{BP09} is taken to be zero).
\end{Example}

\begin{Example}[Poincar\'e ball, the end]\label{Ex_Poin_4}
If the dimension of $M\simeq\mathring B_1\times\R^n\setminus\{0\}$ is bigger or equal
to $4$ and $V$ has compact support, then we know from Example \ref{Ex_Poin_3} that
all the assumptions of Lemma \ref{Hyp_BP} are verified on the open set
$U_\delta\subset\R$. So, Theorem \ref{Thm_Cal} applies, and one has for each
$E\in U_\delta$
$$
\frac\d{\d E}\;\!\mathrm{Cal}\big(\widetilde S_E\big)
=(2E)^{-1/2}\int_{\{(q,p)\in\mathring B_1\times\R^n\setminus\{0\}\mid
|p|^2(1-|q|^2)^2=8E,~p\cdot q=0\}}\Phi(q_+,p^+)\;\!
\frac{\omega^{n-1}(q,p)}{(n-1)\;\!!}\;\!,
$$
with $\Phi(q,p)=\tanh^{-1}\left(\frac{2(p\cdot q)}{|p|(1+|q|^2)}\right)$ and
$(q_+,p^+):=S(q,p)$.
\end{Example}

\section*{Acknowledgements}

R.T.d.A thanks the Institut de math\'ematiques de l'Universit\'e de Neuch\^atel for
its kind hospitality in February and July 2011. A.G. is grateful for the hospitality
provided by the Mathematics Department of the Pontificia Universidad Cat\'olica de
Chile in March 2011.


\def\cprime{$'$} \def\cprime{$'$}

\end{document}